%% file: online_offline.tex
\documentclass[11pt,letterpaper,]{article}
\usepackage{lmodern}
\usepackage{setspace}
\usepackage{amssymb,amsmath}
\usepackage{ifxetex,ifluatex}
\usepackage{fixltx2e} 
\ifnum 0\ifxetex 1\fi\ifluatex 1\fi=0 
  \usepackage[T1]{fontenc}
  \usepackage[utf8]{inputenc}
\else 
  \ifxetex
    \usepackage{mathspec}
  \else
    \usepackage{fontspec}
  \fi
  \defaultfontfeatures{Ligatures=TeX,Scale=MatchLowercase}
\fi
\IfFileExists{upquote.sty}{\usepackage{upquote}}{}
\IfFileExists{microtype.sty}{%
\usepackage{microtype}
\UseMicrotypeSet[protrusion]{basicmath} 
}{}
\usepackage[margin=1.1in]{geometry}
\usepackage{hyperref}
\hypersetup{unicode=true,
            pdftitle={Using an online sample to estimate the size of an offline population},
            pdfkeywords={Demographic methods; Social networks; Sampling and survey research; Digital demography; Computational social science},
            pdfborder={0 0 0},
            breaklinks=true}
\usepackage{longtable,booktabs}
\usepackage{graphicx,grffile}
\makeatletter
\def\maxwidth{\ifdim\Gin@nat@width>\linewidth\linewidth\else\Gin@nat@width\fi}
\def\maxheight{\ifdim\Gin@nat@height>\textheight\textheight\else\Gin@nat@height\fi}
\makeatother
\setkeys{Gin}{width=\maxwidth,height=\maxheight,keepaspectratio}
\IfFileExists{parskip.sty}{%
\usepackage{parskip}
}{
\setlength{\parindent}{0pt}
\setlength{\parskip}{6pt plus 2pt minus 1pt}
}
\providecommand{\tightlist}{%
  \setlength{\itemsep}{0pt}\setlength{\parskip}{0pt}}
\ifx\paragraph\undefined\else
\let\oldparagraph\paragraph
\renewcommand{\paragraph}[1]{\oldparagraph{#1}\mbox{}}
\fi
\ifx\subparagraph\undefined\else
\let\oldsubparagraph\subparagraph
\renewcommand{\subparagraph}[1]{\oldsubparagraph{#1}\mbox{}}
\fi

\let\rmarkdownfootnote\footnote%
\def\footnote{\protect\rmarkdownfootnote}

\usepackage{titling}

  \title{Using an online sample\\
to estimate the size of an offline population\footnote{For helpful feedback on an earlier version of the manuscript, the authors would like to thank Emilio Zagheni and the participants in the pre-PAA 2018 workshop ``Demographic Research in the Digital Age,'' sponsored by IUSSP; and participants in the 2018 Berkeley Demography Seminar Series. We also thank Raphael Nishimura, who helped us find the comparison estimate from Brazil, and two anonymous reviewers.}}
    \pretitle{\vspace{\droptitle}\centering\huge}
  \posttitle{\par}
    \author{\textnormal{Dennis M. Feehan}\footnote{Department of Demography, UC Berkeley, feehan@berkeley.edu}\\
\textnormal{Curtiss Cobb}\footnote{Facebook, Inc.}}
    \preauthor{\centering\large\emph}
  \postauthor{\par}
      \predate{\centering\large\emph}
  \postdate{\par}
    \date{\textnormal{26 June 2019}}

\usepackage{booktabs}
\usepackage{longtable}
\usepackage{array}
\usepackage{multirow}
\usepackage{wrapfig}
\usepackage{float}
\usepackage{colortbl}
\usepackage{pdflscape}
\usepackage{tabu}
\usepackage{threeparttable}
\usepackage{threeparttablex}
\usepackage[normalem]{ulem}
\usepackage{makecell}
\usepackage{xcolor}

\usepackage{rotating}
\usepackage{sidecap}
\usepackage{amsthm}
\usepackage{mathtools}
\usepackage{bbm}
\usepackage{lastpage}
\usepackage{placeins}
\usepackage{subfig}

\begin{document}
\maketitle
\begin{abstract}
Online data sources offer tremendous promise to demography and other social sciences, but researchers worry that the group of people who are represented in online datasets can be different from the general population. We show that by sampling and anonymously interviewing people who are online, researchers can learn about both people who are online and people who are offline. Our approach is based on the insight that people everywhere are connected through in-person social networks, such as kin, friendship, and contact networks. We illustrate how this insight can be used to derive an estimator for tracking the \emph{digital divide} in access to the internet, an increasingly important dimension of population inequality in the modern world. We conducted a large-scale empirical test of our approach, using an online sample to estimate internet adoption in five countries (\(n \approx 15,000\)). Our test embedded a randomized experiment whose results can help design future studies. Our approach could be adapted to many other settings, offering one way to overcome some of the major challenges facing demographers in the information age.
\end{abstract}

\theoremstyle{plain}
\newtheorem{Theorem}{Theorem}[section]
\newtheorem{Result}{Result}[section]
\newtheorem{Corollary}{Corollary}[section]
\theoremstyle{definition}
\newtheorem{Postulate}{Postulate}[section]
\newtheorem{Fact}{Fact}[section]
\newtheorem{Definition}{Definition}[section]
\newtheorem{Conjecture}{Conjecture}[section]
\newtheorem{Problem}{Problem}[section]
\newtheorem{Example}{Example}[section]
\theoremstyle{remark}
\newtheorem{Case}{Case}[section]

\newpage

\hypertarget{introduction}{%
\section{Introduction}\label{introduction}}

Online data sources offer tremendous promise to demography and other social
sciences (Cesare et al. 2018; Lazer et al. 2009; Zagheni and Weber 2012),
but researchers often worry that the group of people who are represented in
online datasets can be different from the general population.
In this study, we develop a strategy for addressing this challenge: we show that by
sampling and anonymously interviewing people who are online, researchers
can learn about both people who are online and people who are offline.

Asking survey respondents to report about others is an idea that has
independently arisen in many different substantive areas
(see, for example, Sirken 1970; Bernard et al. 1991; Hill and Trussell 1977; Marsden 2005).
In demography, the approach can be traced back to Brass and colleagues'
innovative development of census and survey questions that ask respondents
about their parents, spouses, or siblings (Brass 1975).
Our approach can be seen as an extension of this previous work to the situation
where the goal is to learn about everyone in a population, but respondents are
only sampled and interviewed online.
Thus, our study is an illustration of one way to overcome many of
the challenges that face the sampling and survey research community in the
information age.

We illustrate our methodology by developing a new way to study the \emph{digital divide}
in access to the internet around the world.
Scholars use the term digital divide to refer to the fact that
access to the internet is highly unequal:
billions of people around the world
have never been online (Hjort and Poulsen 2019; World Bank 2016);
people in poor countries use the internet much less than people in wealthy
countries (World Bank 2016);
and even within countries that enjoy high levels of internet adoption, research
suggests that access to the internet can differ considerably by age, gender,
income, and race
(Friemel 2016; Haight et al. 2014; Van Deursen and Van Dijk 2014; Vigdor et al. 2014).
Thus, the digital divide is an important dimension of
population inequality in the modern world.

The digital divide is important because research has revealed that access to the
internet may affect health and wellbeing through a wide range of different
mechanisms.
For example, scholars have found that increasing
internet adoption may lead to job creation (Hjort and Poulsen 2019),
improvements in education (Kho et al. 2018),
increases in international trade (Clarke and Wallsten 2006),
increases in social capital (Bauernschuster et al. 2014),
political mobilization (Manacorda and Tesei 2016),
reduced sleep (Billari et al. 2018),
and changes in fertility (Billari et al. 2019).
The World Bank devoted its 2016 World Development Report to the
`digital dividends' that may result from increasing access
to the internet in the developing world (World Bank 2016).

Reliable estimates of internet adoption are typically based on methodologically
rigorous household surveys or censuses (e.g., ICF 2004; Cohen and Adams 2011).
However, this rigor comes at a price: these surveys can be very costly and
typically take months to design and implement (e.g., Parsons et al. 2014; Greenwell and Salentine 2018; ICF 2018; Rojas 2015).
These limitations are especially problematic because internet
adoption appears to be changing on a much faster time-scale than many
conventional indicators of social and economic wellbeing
(Perrin and Duggan 2015; World Bank 2016).

The difficulty of obtaining up-to-date estimates of internet adoption is
unfortunate because researchers need to be able to measure the digital divide in
order to understand its implications for inequality and opportunity;
and policymakers who want to implement and evaluate strategies for making
internet access more widely available rely on being able to measure the level
and rate of change in the number of people who have access to the internet\footnote{For example, the proportion of people using the internet in each country is one of the
  key indicators for the United Nations Sustainable Development Goals; see \href{https://unstats.un.org/sdgs/metadata/}{SDG indicator} 17.8.1.}.

To help address this challenge, we use our methodology to develop an alternative approach
to estimating internet adoption that is dramatically faster and cheaper than
conventional surveys: we interviewed a sample of Facebook users and asked them
whether or not members of their offline personal networks use the internet.
Our approach is based on the insight that internet users
are connected to many other people through in-person social networks such as
kin, friendship, and contact networks. By interviewing a sample of
Facebook users and anonymously asking about the members of these offline social networks,
we can learn about both people who are online and people who are
not online.

\hypertarget{sec:setup}{%
\section{Methods}\label{sec:setup}}

People everywhere are connected to one another through kinship, friendship,
professional activities and interpersonal interactions.
Our strategy for obtaining fast and inexpensive estimates of internet adoption
is based on asking people sampled online to report about internet adoption among other
people they are connected to in these everyday, offline personal networks.
The challenge is to determine how to turn people's anonymous reports about their
personal network members into estimates of internet adoption.
We now explain how we used a formal framework called network reporting
to understand which quantities we need to estimate in order to accomplish our goal
(Feehan 2015; Feehan and Salganik 2016a).
(A detailed derivation can be found in Appendix \ref{sec:estimator}.)

\begin{figure*}
    \centering
        \subfloat[\label{fig:nr01}]{\includegraphics[width=1.5in]{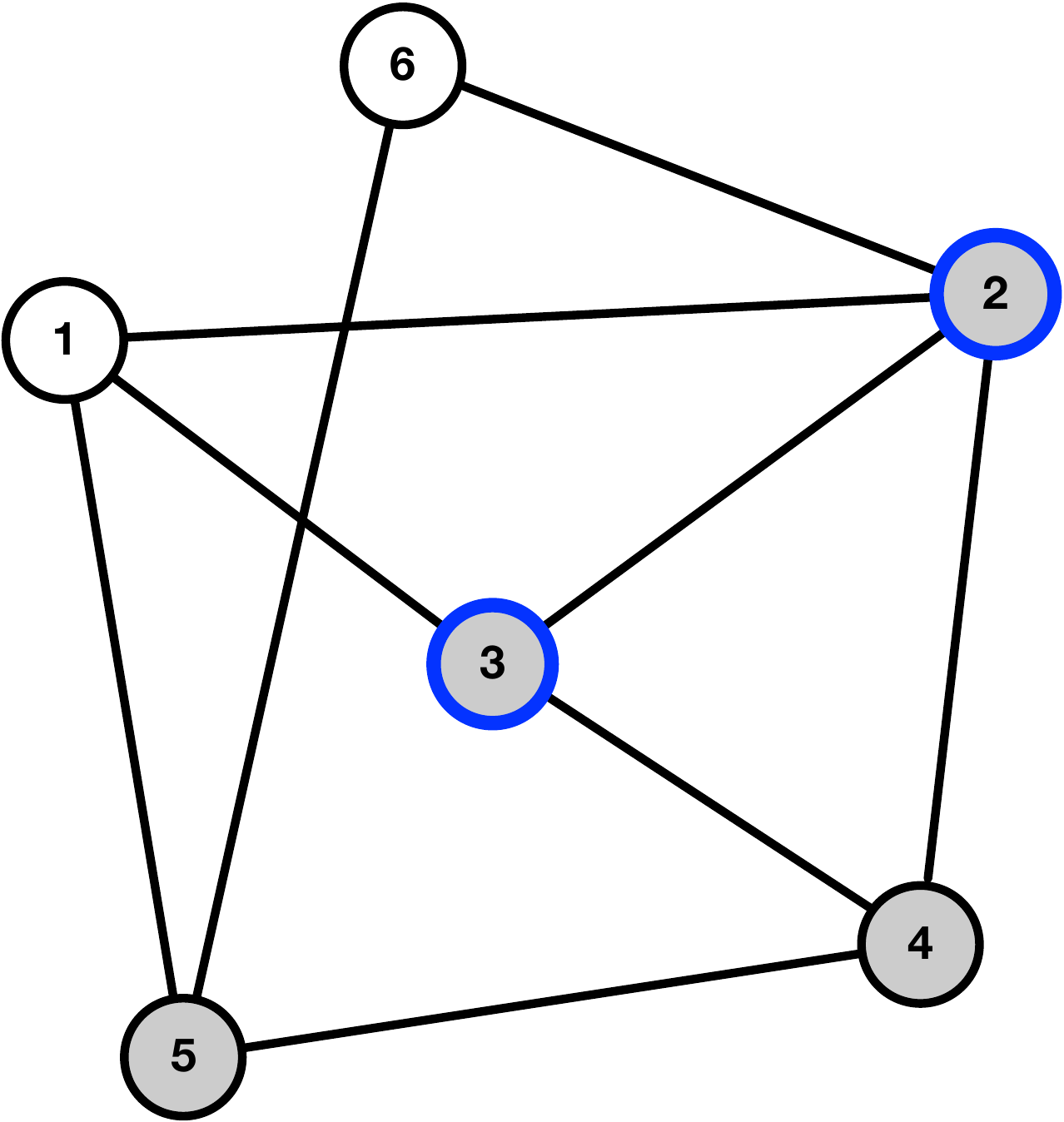}}
        \hspace{5mm}%
        \subfloat[\label{fig:nr02}]{\includegraphics[width=1.5in]{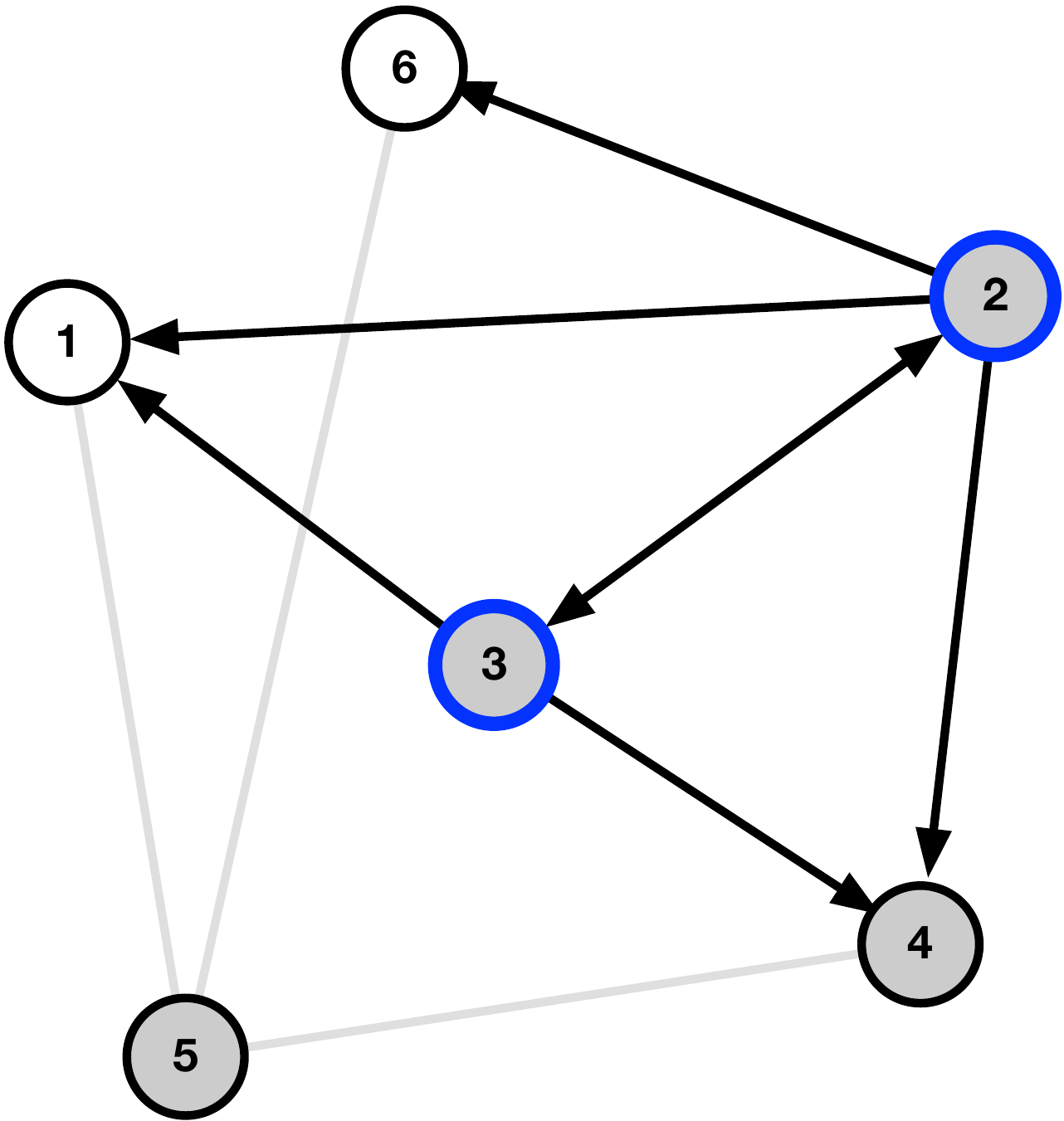}}
        \hspace{5mm}%
        \subfloat[\label{fig:nr03}]{\includegraphics[width=1in]{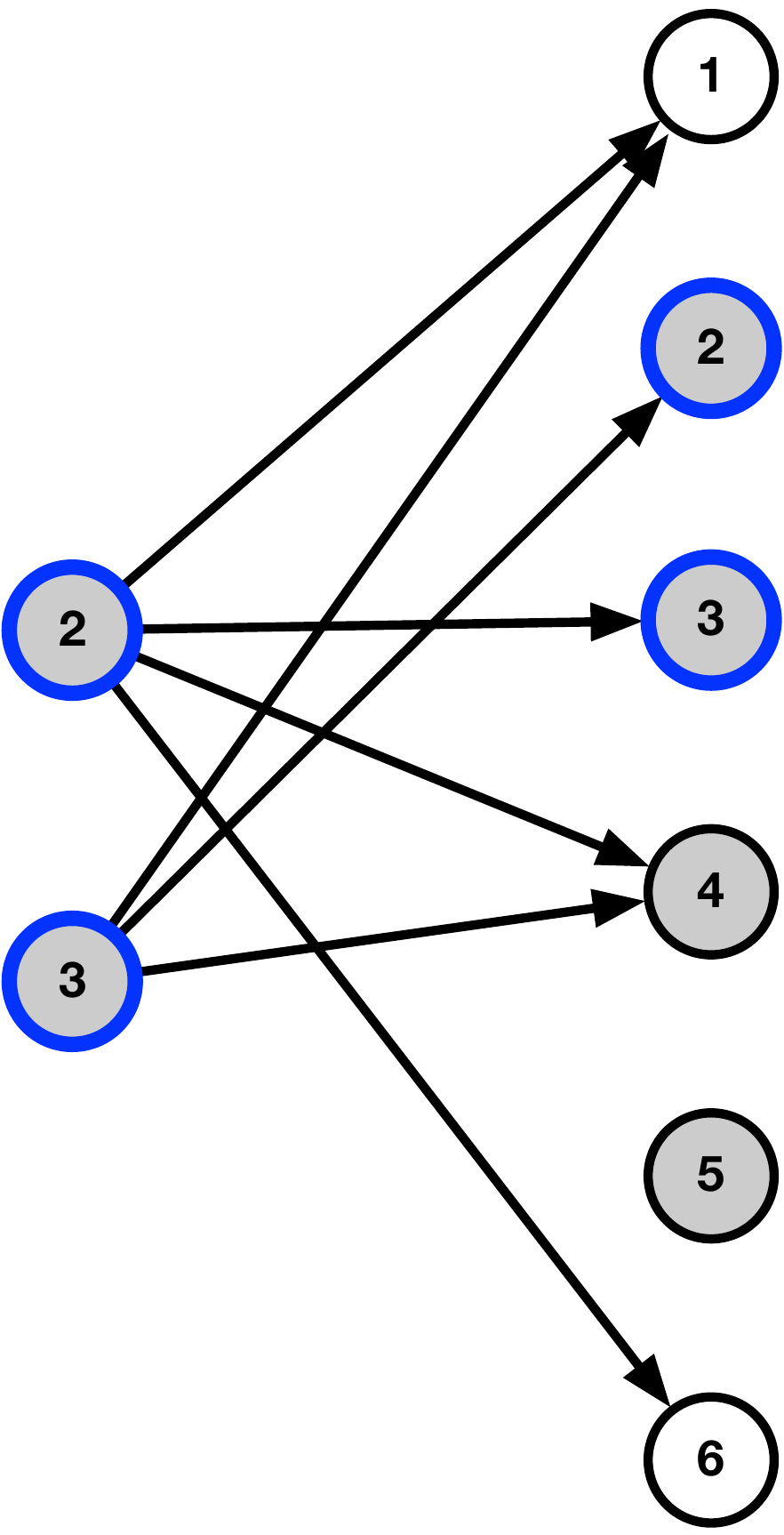}}
        \caption{Network reporting setup: asking people on Facebook to report about their offline personal networks.}
        \label{fig:nr-illustration}
\end{figure*}

Fig. \ref{fig:nr-illustration} illustrates the general setup with an example.
Fig. \ref{fig:nr01} shows six people who are connected together in a social network.
The network relation is symmetric, meaning that whenever person A is connected to person B,
then B is also connected to A.
We make a distinction between nodes that can potentially be sampled and
interviewed---the \emph{frame population}---and other nodes.
For example, a frame population might be cell phone users;
the users of a specific app such as Facebook; or
people who live at addresses that can be reached by postal mail.
In Fig. \ref{fig:nr-illustration}, nodes 2 and 3 are in the frame population.

Fig. \ref{fig:nr02} shows the \emph{reporting network} that is generated when both nodes 2 and 3
are interviewed about the people they are connected to in the social network.
The reporting network is different from the social network:
the social network has an undirected edge \(A-B\) when A and B are
socially connected;
the reporting network, on the other hand, has a directed edge
\(A \rightarrow B\) whenever A reports about B.
When reporting is accurate, there will be structural similarities between the
social network and the reporting network, but this need not be true in general.
The reporting network is a useful formalism that can help researchers
develop estimators, understand possible sources of reporting errors,
and derive self-consistency checks.

Fig. \ref{fig:nr03} shows a rearrangement of Fig. \ref{fig:nr02} that is helpful for deriving
estimators from a reporting network.
On the left-hand side of Fig. \ref{fig:nr03} is the set of nodes that makes
reports (the frame population), and on the right hand
side is the set of nodes that can be reported about (the universe)\footnote{Note that a particular node can appear in both sides if it is in the frame population and in the universe.}.
Drawn this way, every report must connect a node on the left-hand side to a node on
the right-hand side.
Thus, the total number of reports that leaves the left-hand
side must equal the total number of reports that arrives at the right-hand side.
Mathematically, this means that when everyone in the frame population is interviewed,
we have the following identity:

\begin{equation}
\text{\# internet users} = N_H =
 \frac{
    \overbrace{y^{+}_{F,H}}^{\mathclap{\substack{\text{\# reported connections from} \\ \text{people on FB leading to internet users}}}}
 }{
    \underbrace{\bar{v}_{H, F}}_{\mathclap{\substack{\text{average number of times each} \\ \text{internet user gets reported}}}} 
 } 
\label{eq:nrid-text}
\end{equation}

The denominator of Eq. \ref{eq:nrid-text} is a quantity called the \emph{visibility} of
internet users.
The visibility is the number of times the average internet
user would get reported in a census of the frame population.
Intuitively, Eq.~\ref{eq:nrid-text} divides by the visibility to adjust for the fact
that the average internet user would be reported multiple times in a
census of the frame population.

\hypertarget{instrument-design}{%
\subsubsection*{Instrument design}\label{instrument-design}}
\addcontentsline{toc}{subsubsection}{Instrument design}

In principle, people can be asked to report about
any type of personal network relationship that is symmetric.
Thus, the specific type of personal network that respondents are asked to report
about---the \emph{tie definition}---is a study design parameter that researchers are free to vary
(Feehan et al. 2016). To explore the impact of this study design parameter,
we embedded a randomized experiment in our surey.
In our experiment, survey respondents were randomly assigned to report about one of two different tie
definitions: the meal tie definition and the conversational contact tie
definition (Tbl.~\ref{tbl:tie-definitions}).
We chose these two tie definitions because
(1) previous research led us to believe that respondents can plausibly report
the number of people that they interacted with in the previous day, avoiding
the need to indirectly estimate personal network sizes; (2) researchers have
had success using versions of these tie definitions in previous studies
(Feehan et al. 2016; Mossong et al. 2008).

\begin{table}[t]

\caption{\label{tab:unnamed-chunk-2}The two different networks survey respondents were asked about. In our survey experiment, respondents were randomly assigned to report about one of these two networks. \label{tbl:tie-definitions}}
\centering
\begin{tabular}{>{\raggedright\arraybackslash}p{16em}|>{\raggedright\arraybackslash}p{16em}}
\hline
Meal network & Conversational contact network\\
\hline
How many people did you share food or drink with yesterday? These people could be family members, neighbors, or other people. Please include all food or drink taken at any location, including at home, at work, at a cafe, or in a restaurant. & How many people did you have conversational contact with yesterday? By conversational contact, we mean anyone you spoke with face to face for at least three words.\\
\hline
\end{tabular}
\end{table}

Each survey interview took place in two phases:
in the first phase, survey respondents were asked to report the size of their
personal networks
(e.g., ``How many people did you share food or drink with yesterday?''; Tbl.~\ref{tbl:tie-definitions}).
In the second phase, the goal was to obtain information about internet use among
the members of each respondent's personal network.
Ideally, the respondent would provide information about every single person in her network
one by one.
However, this approach seemed likely to produce unacceptable levels of
respondent fatigue (Eckman et al. 2014; Tourangeau et al. 2015).
Therefore, in the second phase of the interview respondents were asked for information about
the three members of their personal networks who `came to mind' first (Fig. \ref{fig:alters}).
We call these people that we obtain additional information about \emph{detailed alters}\footnote{We did not ask for any sensitive or personally identifying information about
  these three detailed alters.}.
(Additional details and our survey instrument are included in Appendix \ref{sec:survey-methods}.

\hypertarget{estimators}{%
\subsubsection*{Estimators}\label{estimators}}
\addcontentsline{toc}{subsubsection}{Estimators}

The identity in Eq. \ref{eq:nrid-text} would hold if we obtained a census of
monthly active Facebook users.
In practice, we have a sample and not a census; therefore,
we construct an estimator for the number of internet users by developing sample-based
estimators for the numerator and the denominator of Eq. \ref{eq:nrid-text}.
We now describe these two components in more detail.

Given information about respondents' network sizes and the detailed alters' internet
use, the numerator of \ref{eq:nrid-text} (\(y_{F,H}\)) can be estimated from our sample with:

\begin{equation}
\widehat{y}_{F,H} = 
\sum_{i \in s} w_i \frac{d_i}{r_i} o_i,
\label{eq:estimator-numerator}
\end{equation}

\noindent where \(s\) is the sample of Facebook users;
\(w_i\) is the expansion weight for \(i \in s\);
\(d_i\) is the network size (degree) of \(i \in s\);
\(r_i\) is the number of detailed alters from \(i \in s\) (\(r_i \in \{1, 2, 3\}\));
and \(o_i\) is the number of detailed alters reported to be online.

We calculate \(w_i\) by approximating our design as a as a simple random sample,
post-stratified by age and gender.
(Appendix \ref{sec:survey-methods} has more information on our weighting.)
In order to use information about the \(r_i\) detailed alters to make inferences about the \(d_i\)
people in the respondent's network, the estimator in Eq. \ref{eq:estimator-numerator}
makes the additional assumption that the detailed alters are a simple random sample
of respondents' personal networks.
Thus, \(\frac{d_i}{r_i}\) can be seen as a weight that accounts for sampling
\(r_i\) out of the \(d_i\) personal network members.
Previous work on egocentric survey research suggests that, instead of being
a simple random sample, network members who
come to mind first may be more likely to come from the same social context,
and may be more likely to be strongly connected to the respondent (Marsden 2005).
Therefore, we develop two different ways to assess this assumption:
first, we introduce internal consistency checks that can detect systematic
biases that would emerge if detailed alters are very different from other personal network
members (sec.~\ref{sec:ic-checks});
and, second, we introduce a sensitivity framework which enables us to formally assess the
impact that different magnitudes of selection bias among the detailed alters would
have on our estimates (Appendix \ref{sec:sensitivity-analysis}).

\begin{figure*}
        \includegraphics{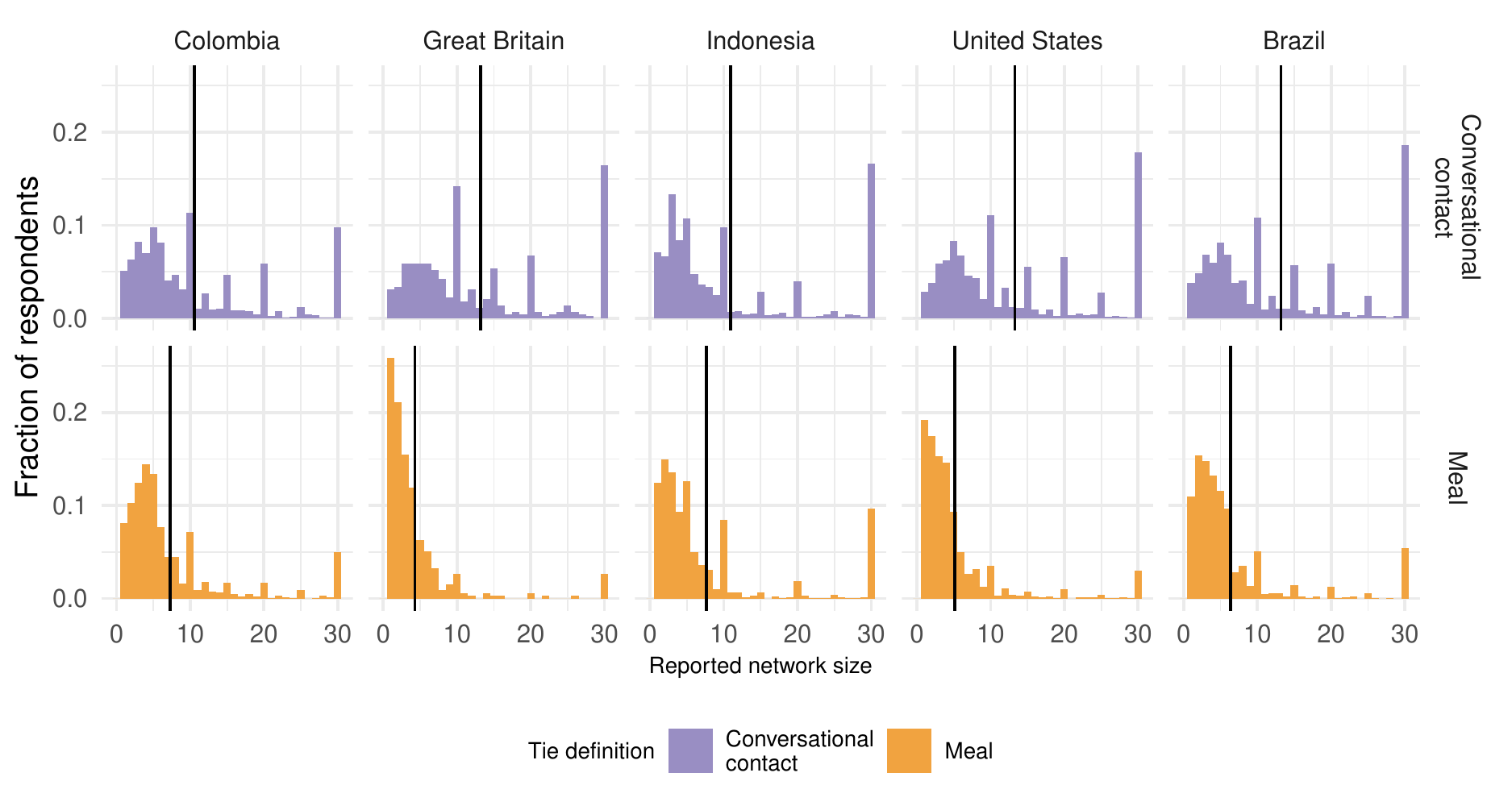}
    \caption{
    Estimated degree distributions for the conversational contact network (left)
and the meal network (right). The vertical line on each panel shows the average.
Average personal network size is smaller for the
meal network than for the contact network; further, the contact network shows
greater evidence of heaping on multiples of 5 and 10 than the meal network. 
These findings are consistent with a hypothesized tradeoff between the quality
and the quantity of information reported in personal
networks. Responses higher than 30 are coded as 30 in these plots.
    }
    \label{fig:estimated-degrees}
\end{figure*}

The denominator of Eq. \ref{eq:nrid-text} (\(\bar{v}_{H,F}\)) is a quantity called
the \emph{visibility} of
internet users, which is defined as the number of times that the average
internet user would be reported in a census of active Facebook users.
Many different strategies could be used to estimate or approximate the
visibility of internet users;
here, we adopt a simple approach: we use the average number of times that
a Facebook user shares a meal with another Facebook user
to approximate the visibility of internet users.
Mathematically, this assumption can be written

\begin{equation}
\bar{d}_{H, F} = \bar{d}_{F,F}.
\label{eq:equal-tie-fbnofb}
\end{equation}

\noindent The condition in Eq. \ref{eq:equal-tie-fbnofb} requires that
two quantities be equal: (1) the rate at which someone who is on the internet
shares a meal with someone who is on Facebook (\(\bar{d}_{H,F}\));
and, (2) the rate at which someone who is on Facebook shares a meal with
someone who is also on Facebook (\(\bar{d}_{F,F}\)).
This assumption would hold if, for example, people who are on the internet
do not pay attention to whether or not another internet user is on
Facebook when deciding to share a meal together.
This assumption could be violated if, for example, people frequently
organize sharing a meal together using Facebook without inviting other
people.
We explore how violating this condition affects estimates as part of
a sensitivity analysis in Appendix \ref{sec:sensitivity-analysis};
in Appendix \ref{sec:simple-model} we develop a simple model that motivates
this condition;
and in Sec.~\ref{sec:conclusion}, we discuss how additional data collection could
remove the need for this condition altogether.

Given the condition in Eq. \ref{eq:equal-tie-fbnofb},
we can estimate \(\bar{v}_{H,F}\) with an estimator for \(\bar{d}_{F,F}\),
the average number of meals that someone on Facebook reports sharing with
someone else on Facebook.
In order to estimate \(\bar{d}_{F,F}\), we use

\begin{equation}
\widehat{\bar{d}}_{F,F} = \frac{\sum_{i \in s} w_i \frac{d_i}{r_i} f_i}{\sum_{i \in s} w_i},
\label{eq:degree-estimator}
\end{equation}

\noindent where the new quantity, \(f_i\), is the number of Facebook users that
respondent \(i\) reports among her detailed alters.

Putting Eq. \ref{eq:estimator-numerator} and Eq. \ref{eq:degree-estimator} together, we have

\begin{equation}
\begin{aligned}
\widehat{N}_H &= \frac{\widehat{y}_{F,H}}{\widehat{\bar{d}}_{F,F}}
= 
\frac{\sum_{i \in s} w_i \frac{d_i}{r_i} o_i}
     {\sum_{i \in s} w_i \frac{d_i}{r_i} f_i}\times
\sum_{i \in s} w_i.
\label{eqn:estimator-any-weights}
\end{aligned}
\end{equation}

Appendix \ref{sec:estimator} has a detailed derivation of the estimator and
a precise description of all of the conditions it relies upon;
Appendix \ref{sec:generalized} describes an alternate approach to producing
estimates using data we collected; and
Appendix \ref{sec:sensitivity-analysis} has a framework for sensitivity analysis
which can be used to understand how estimates are affected by violations of these
conditions.

\hypertarget{sec:results}{%
\section{Results}\label{sec:results}}

We used Facebook's survey infrastructure to obtain a simple random sample of
people who actively use Facebook in five countries around the world:
Brazil (n=3,761),
Colombia (n=4,157),
Great Britain (n=781),
Indonesia (n=2,794), and
the United States (n=4,288)\footnote{We consider users to be active if they have logged onto Facebook in the 30 days before the survey; we also restrict responses to people over 15 years old.}.
We chose these countries because they span a breadth of expected levels of internet adoption
and economic development.
Respondents were slightly more female than male in all countries except for Indonesia,
and age distributions are typical of monthly active Facebook users in these countries.
Figure \ref{fig:respondents} shows the age and gender distribution of survey respondents for
each tie definition\footnote{In order to ensure that the survey instrument and methods worked well, we started with
  a smaller sample in Great Britain (which is why there are fewer respondents in that country).}.
All estimates below are weighted to account for the sample design and to be representative
of the universe of monthly active Facebook users in each country.
Estimates of sampling uncertainty are based on the rescaled bootstrap method
(Feehan and Salganik 2016b; Rao and Wu 1988; Rao et al. 1992).

\begin{figure*}[p]
        \includegraphics[]{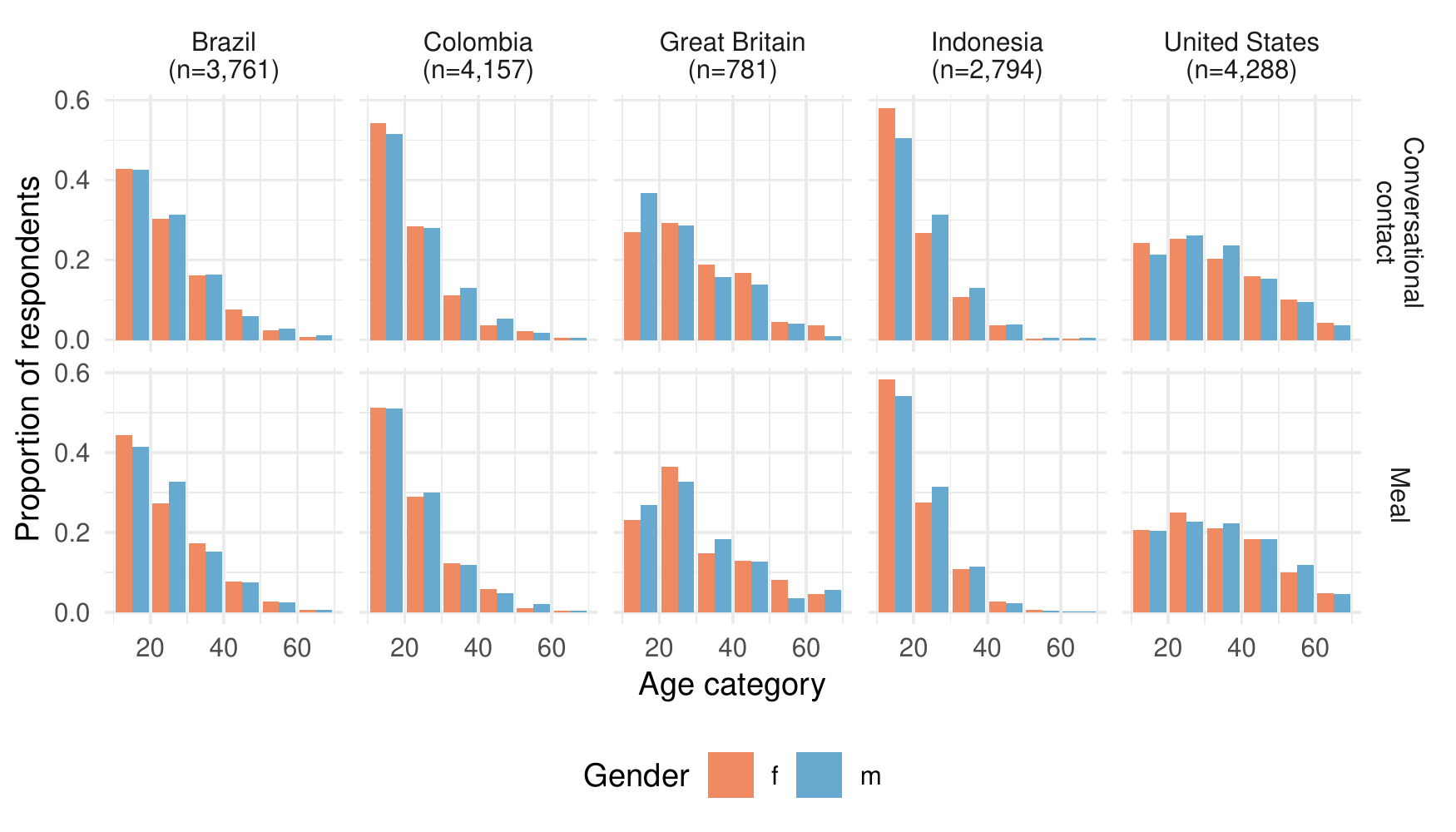}
        \caption{Age and gender of survey respondents in each country. Estimates throughout this article
             use sampling weights to account for sample design and nonresponse.}
    \label{fig:respondents}
\end{figure*}

Figure \ref{fig:estimated-degrees} shows the distribution of personal network sizes
reported by respondents from each country, and for each tie definition\footnote{Recall that respondents were randomly assigned to report either about meal networks or
  about conversational contact networks; thus, Fig. \ref{fig:estimated-degrees} and all
  subsequent figures show results broken down by tie definition.}.
The average size of meal networks was smaller than conversational contact networks
in all countries (Table \ref{tbl:raw_degree}): the average reported size of the
meal network varied from about 4 (Great Britain) to about 8 (Indonesia), while
the average reported size of the conversational contact network varied from
about 11 (Colombia and Indonesia) to about 13 (Brazil, Great Britain,
and the United States).
For both networks, Fig. \ref{fig:estimated-degrees} suggests that there may be heaping in
reported network sizes that are multiples of five and ten; this heaping is more
evident in the reported number of conversational contacts than for meals,
suggesting that reports about the meal network may more accurate.

\hypertarget{sec:ic-checks}{%
\subsection{Internal consistency checks}\label{sec:ic-checks}}

\begin{figure*}[p]
        \includegraphics{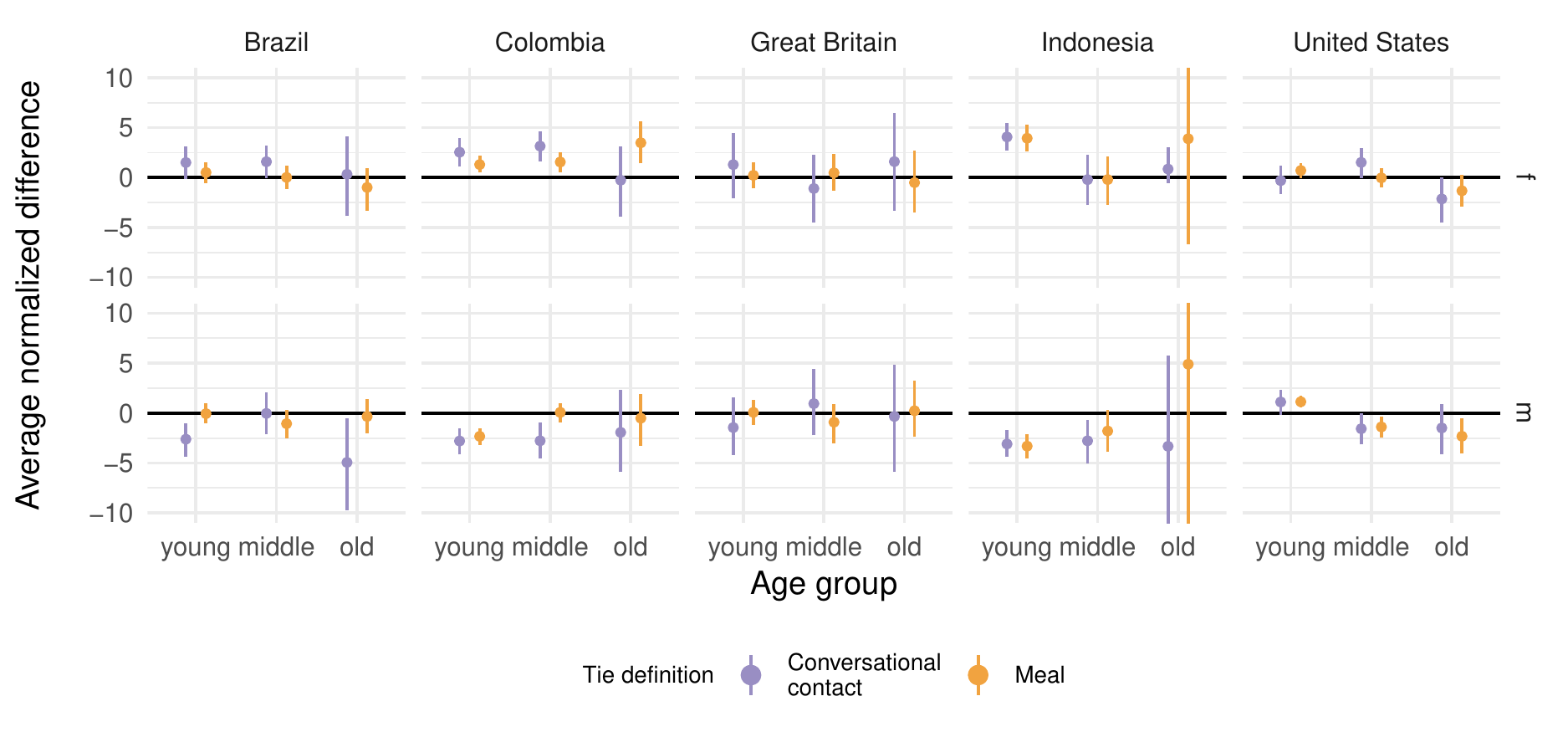}
    \caption{
    Internal consistency checks. 
By estimating the same quantity using independent parts of our sample, we can assess the internal consistency of respondents' network reports.
Estimated difference between two independent estimates of the same quantity and 
95\% confidence intervals are shown for each age-gender group and each type of network; 
an estimate of 0 means that the two independent estimates are exactly the same. 
Across most age-sex groups, results are internally consistent with one another, within sampling error; however, some groups show evidence of reporting errors (e.g. young people in Indonesia). 
Results also suggest that reports about the meal definition are more internally consistent, 
even though meal networks are smaller than conversational contact networks.
    }
    \label{fig:ic}
\end{figure*}

In order to more formally assess the accuracy of reports about each network,
we develop \emph{internal consistency checks} (Bernard et al. 2010; Brewer et al. 2000; Feehan et al. 2016).
These internal consistency checks use the information about the age group and gender
of each detailed alter that respondents reported about.
The idea is to find reported quantities that can be estimated from the
data in two different ways.
To the extent that these independent estimates of the same quantity agree,
the reported network connections are internally consistent.
For example, using survey responses from only men, we can estimate the number
of connections between men and women; similarly, using survey responses from
only women, we can estimate the number of connections between women and men.
By definition, these two quantities are equal;
thus, under perfect conditions where our survey
does not suffer from any reporting errors or selection biases, we would expect
these two independent estimates to agree (up to sampling noise).

We devised internal consistency checks based on reported connections to
and from each of twelve different age-sex groups, by country and by tie definition.
For each age-sex group \(\alpha\), we estimate the average number of connections from
Facebook users in age-sex group \(\alpha\) to Facebook users not in \(\alpha\) (\(d_{F_\alpha, F_-\alpha}\)).
We also estimate the average number of connections from
Facebook users not in age-sex group \(\alpha\) to Facebook users who are in
age-sex group \(\alpha\) (\(d_{F_-\alpha, F_\alpha}\)).
We then define the average normalized difference \(\Delta_\alpha\) to be

\begin{equation}
\Delta_{\alpha} =
K \left( \widehat{d}_{F_{-\alpha}, F_{\alpha}} - \widehat{d}_{F_{\alpha}, F_{-\alpha}} \right),
\label{eq:ic-delta-defn}
\end{equation}

\noindent where \(K\) is a scaling factor that is intended to make it easier to compare
different countries and age-sex groups (Appendix \ref{sec:internal-consistency}).
In the absence of any reporting error, selection biases, or sampling variation,
we would expect \(\Delta_\alpha = 0\).
On the other hand,
if there is homophilic selection bias in the respondents' choice of detailed alters
or if members of group \(\alpha\) are especially conspicuous then we expect \(\Delta_\alpha > 0\);
similarly, if there is heterophilic selection bias in respondents'
choice of detailed alters or if members of group \(\alpha\) are especially inconspicuous,
then we expect \(\Delta_\alpha < 0\).

Fig. \ref{fig:ic} shows the average normalized difference (\(\Delta_\alpha\)) for internal
consistency checks based on reported connections to and from each of twelve
different age-sex groups, by country and by tie definition.
Several notable features emerge from Fig. \ref{fig:ic}.
First, for many of the internal consistency checks, the averaged normalized differences
are close to zero, or have confidence intervals that contain zero.
Second, Fig. \ref{fig:ic} suggests that reports based on the meal network are, on average,
more internally consistent than reports based on conversational contact (confirmed in Appendix \ref{sec:additional-results}).
Third, there appears to be no universal pattern that describes deviations in internal consistency checks that are not close to zero.
Taking the example of Indonesia, the average normalized differences for younger age groups suggest
that young women may be relatively conspicuous or that young women are relatively homophilous\footnote{Conspicuousness and homophilic reporting are not distinguishable from the data. In this discussion, we focus on conspicuousness; however, instead of Indonesian women being conspicuous, it could also be the case that Indonesian women have homophilic selection biases in choosing their
  detailed alters (i.e., they tend to report other women at a higher rate than would be expected from simple random sampling of their network members).}.
On the other hand, young men are relatively inconspicuous or relatively heterophilous.
In Brazil and Colombia, similar patterns appear for the conversational contact network.
In Great Britain and the United States, however, most of the IC checks suggest that reports
are internally consistent.

\hypertarget{comparing-tie-definition-accuracy}{%
\subsection{Comparing tie definition accuracy}\label{comparing-tie-definition-accuracy}}

Fig. \ref{fig:ic-diff} directly compares the difference in internal consistency results
for the conversational contact and meal networks.
The figure shows the estimated sampling distrubution of \(\text{TAE}\), the total
absolute difference between the internal consistancy checks for the
conversational contact network and the internal consistency checks for the meal
network:

\begin{equation}
\text{TAE} = \sum_\alpha \left(|\Delta_{\alpha,\text{cc}}| - |\Delta_{\alpha,\text{meal}}| \right),
\end{equation}

\noindent where \(|\Delta_{\alpha,\text{cc}}|\) and \(|\Delta_{\alpha,\text{meal}}|\) are the absolute
internal consistency check statistics based on group \(\alpha\) for the conversational
contact and meal networks (i.e., the absolute value of Eq. \ref{eq:ic-delta-defn}).
Thus, TAE is a summary of how well the internal consistency checks perform across all
age-sex groups for the conversational contact network minus the meal network.
Since values of \(|\Delta_\alpha|\) close to 0 indicate more internally consistent reports,
when TAE is positive, that suggests that the meal network is more internally consistent;
conversely, when TAE is negative, that suggests that the conversational contact network is
more internally consistent.
For all countries except for Indonesia, the majority of the mass of the
estimated distribution is greater than 0, suggesting that the meal network
reports are more internally consistent than conversational contact network
reports (Tbl.~\ref{tbl:ic_err}).

\begin{figure*}
  \centering
  \includegraphics[width=3in]{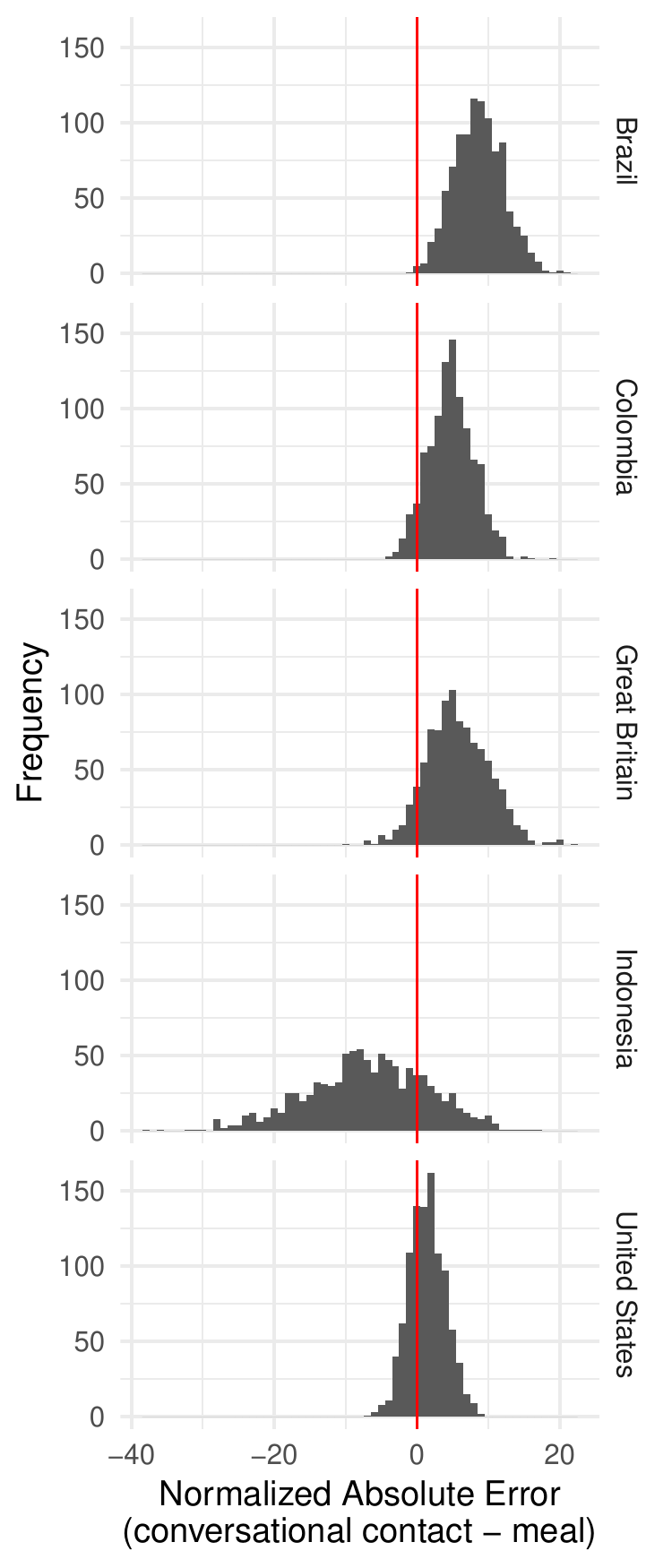}
  \caption{Estimated sampling distribution of the difference between the normalized square error for internal consistency checks from the conversational contact network and from the meal network. For all countries except for Indonesia, the meal network was more internally consistent than the conversational contact network (Table \ref{tbl:ic_err}).}
  \label{fig:ic-diff}
\end{figure*}

\hypertarget{estimates-of-internet-adoption}{%
\subsubsection*{Estimates of internet adoption}\label{estimates-of-internet-adoption}}
\addcontentsline{toc}{subsubsection}{Estimates of internet adoption}

\begin{figure}[p]
    \centering
        \includegraphics[]{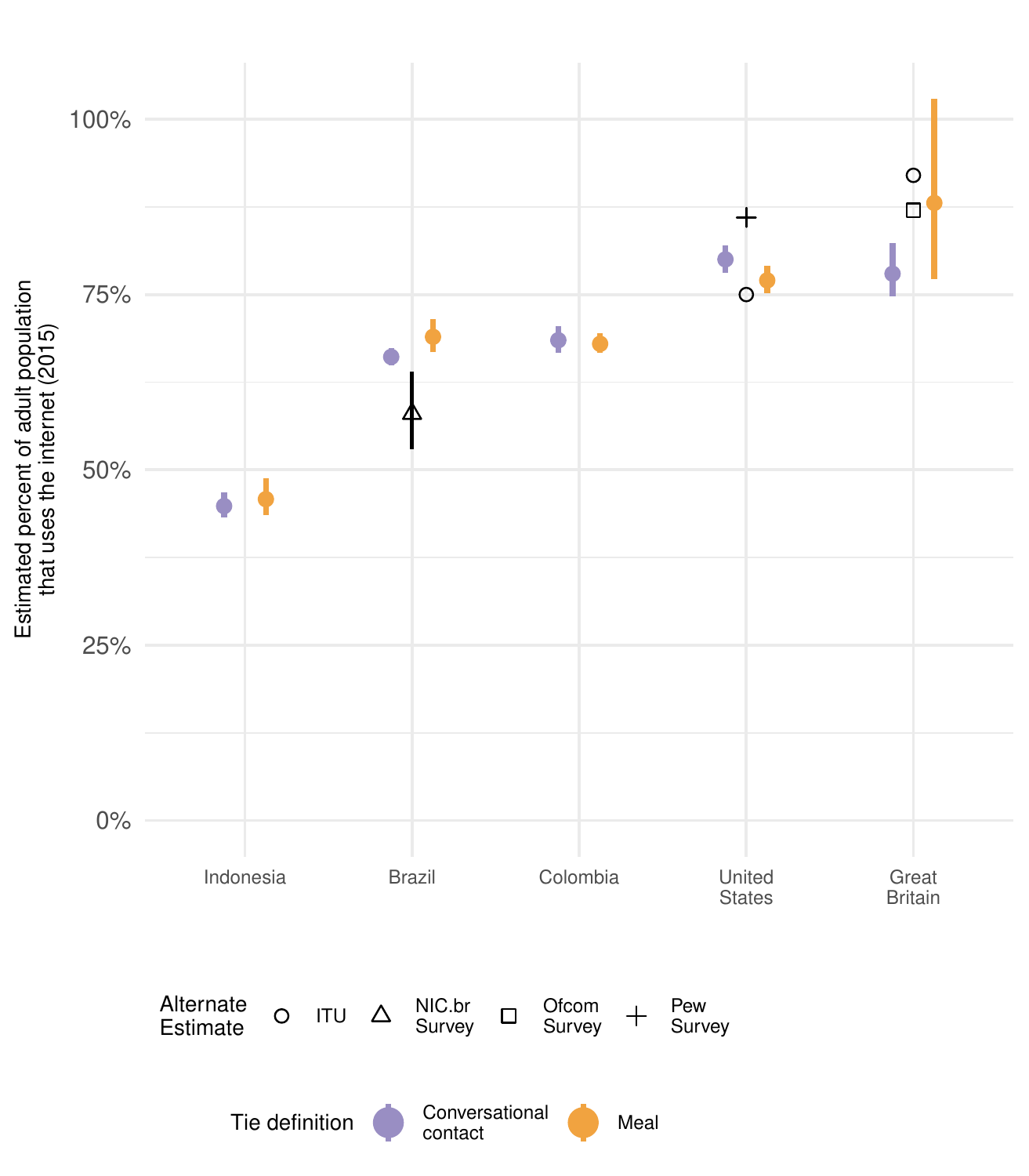}
        \caption{
        Estimated percentage of 2015 adult population that uses the internet, by country and for each of the two networks. 95\% confidence intervals are based on the estimated sampling distribution from the rescaled bootstrap.
        For comparison, estimates from alternate sources are shown where available.
In Great Britain, comparison estimates are available the ITU and from an Ofcom
survey; in the United States, comparison estimates are available for
from the ITU and from Pew surveys; and in Brazil, comparison estimates are available for a NIC.br Survey.
The confidence interval around the Brazil estimate is based on the survey's published margin of error. 
        }
        \label{fig:estimates}
\end{figure}

Fig. \ref{fig:estimates} shows estimated internet adoption for each country in our
sample, using each tie definition\footnote{For the purposes of this study, we say that a person has adopted the internet if she has used the internet on a computer or a phone in the last 30 days; Appendix \ref{sec:survey-methods} shows our survey instrument.}.
Two findings emerge from Fig. \ref{fig:estimates}.
First, estimated internet adoption rates are very similar for the conversational
contact and for the meal networks; in all countries, the confidence intervals
for estimates from the two tie definitions overlap.
Second, the countries can be divided into three groups according to estimated
adoption rates:
the United States and Great Britain have the highest rates of internet adoption (above 75\%);
Brazil and Colombia have estimated internet adoption rates between 50\% and 75\%;
and Indonesia has estimated adoption rates below 50\%.
This ordering is consistent with what would be predicted if economic factors such as
GDP per capita were the main driver of internet adoption.

Ideally, we would evaluate our estimator by comparing it to gold standard
measurements of internet adoption in each of the five countries.
Unfortunately, no such gold standard currently exists.
Therefore, in order to further assess the plausibility of the estimates presented in
Fig. \ref{fig:estimates}, we compared our results to existing internet adoption
estimates for Great Britain, the United States, and Brazil, the countries where high-quality alternative
estimates were available\footnote{Our comparisons come from the Pew Research Center
  (Pew Research Center 2018), which is based on a national phone survey in
  the US; from an Ofcom Survey in the UK (Ofcom 2016); from estimates
  reported by the International Telecommunications Union (ITU 2018);
  and from a household survey conducted by NIC.br in Brazil (NIC.br 2016).
  Note that the ITU estimate for the US has all people over aged 3 in the
  denominator and the NIC estimate for Brazil has all people over aged 10 in the
  denominator. All other estimates are for adults.}.
The results show that the fast and inexpensive network reporting estimates
are within the range of other estimates in the United States,
similar to or slightly lower than other estimates in Great Britain,
and somewhat higher than the other estimate for Brazil.

\hypertarget{sec:results-summary}{%
\subsection*{Summary and discussion}\label{sec:results-summary}}
\addcontentsline{toc}{subsection}{Summary and discussion}

We found that estimates of internet adoption from the two different networks
were very similar (Fig. \ref{fig:estimates}). We could not validate our estimates by
comparing them to gold-standard measurements of internet adoption rates because
such a gold standard was not available.
However, a comparison to high-quality alternative estimates in the United
States, Great Britain, and Brazil showed that the network reporting estimates are
consistent with other sources of estimates in the United States, slightly
higher than the other estimate for Brazil, and consistent
or slightly lower than other estimates from Great Britain (Fig. \ref{fig:estimates}).
Thus, we conclude that our fast and inexpensive strategy for obtaining approximate
estimates of internet adoption is promising.

We also found that
(1) in all five countries, reports from the stronger network tie, meals, produced information about fewer
people than the weaker network tie, conversational contact; but,
(2) reports from the stronger network tie produced, on average, more accurate
information than reports from the weaker tie in all countries except for Indonesia (Fig. \ref{fig:ic-diff}).
This finding is consistent with a hypothesized trade-off between the quantity
and quality of information produced by network reports (Feehan et al. 2016);
previous work found support for this theory in network reports about
interactions in the 12 months before the interview.
We find that this tie strength trade-off may operate even when reports are
about interactions that took place the day before the interview.
Future research could compare different time-windows to see if the hypothesized
tradeoff between the quantity and quality of information operates across time
within a fixed type of network tie.
Over time, we hope that a deeper understanding of the relationship between reporting
accuracy and the different dimensions of network tie definitions will accumulate,
leading to useful guidance about how to design studies like ours.

The internal consistency checks suggest that people's reports about their
network members can suffer from reporting errors, and that these reporting
errors vary by who is being reported about (Fig. \ref{fig:ic}).
One possible mechanism for this result could be differential salience of interactions;
another possible mechanism could be homophilic selection of the detailed alters.
This phenomenon is important to understand for measurement, and scientifically
interesting in its own right;
future research could explore different study designs to try and distinguish between
the salience of different demographic groups on the one hand and
selection bias among the detailed alters on the other.
More generally, the internal consistency checks provide a way to evaluate the
quality of reporting from different survey designs, enabling researchers
to experiment with new designs each time data are collected.
Over time, this process may help discover tie definitions
that minimize reporting error (Feehan et al. 2016).

\hypertarget{sec:conclusion}{%
\section{Conclusion}\label{sec:conclusion}}

We showed that a sample of people who are online can be used to estimate
characteristics of a population that is not entirely online.
Our approach is based on the idea that people who are sampled online can
be asked to anonymously report things about other people
to whom they are connected through different kinds of personal networks.
We illustrated our approach by estimating internet adoption in five different
countries around the world.
Our study included a survey experiment that can help inform future efforts
to use online samples to estimate population characteristics.

Our results suggest several possible avenues for future work.
In this study, we focused on simple, design-based estimators.
A natural next step would be to start to build more complex models
using these data.
These models could exploit the relationships that are embedded in the
internal consistency checks as a kind of constraint, estimating adjustments
to ensure that reports are internally consistent.
Such a model could potentially improve the accuracy of the resulting
estimates.
A second next step would be to use our approach to produce estimates
of internet adoption by age and gender.
Finally, future work could explore the possibility of an even simpler estimator
based on asking each respondent about aggregate connections to people who use
the internet (e.g.~``How many of your network members use the internet?'';
Bernard et al. (2010)).
This approach would forgo the ability to conduct internal consistency checks and
to produce estimates by age and gender, but it would be even simpler and shorter than
the approach we used here.

We view our method as a complement to other promising approaches to producing
population-level estimates using online samples.
For example, one stream of research focuses on using changes over time among
members of the online sample to estimate population changes; this approach can
be useful for studying topics like migration
(e.g., Zagheni and Weber 2012).
A second stream of research uses models
that relate people in the online sample to the general population
using covariate information observed in both sources
(e.g., Goel et al. 2015; Fatehkia et al. 2018).
We expect sampling and interviewing people about members of their
offline networks will be especially promising in situations where few or
no people in the group being studied can be expected to be in the online
sample; but we also expect that there will be situations where these alternatives
are more apprporiate than network reporting.
As the field of digital demography emerges, it will be important deepen our
understanding of the trade-offs between these approaches, and to continue to
develop new methods for producing population estimates from an online sample.

We also see our approach as a complement, rather than a replacement for conventional surveys.
The ideal situation would combine frequent, inexpensive estimates,
such as the ones described here, with less frequent conventional surveys.
For example, a conventional probability sample of the general population in a country
could be used to empirically estimate the average number of meals shared between
an internet user and a Facebook user; with direct estimates of that quantity,
the need for a key assumption in our estimator could be completely removed.
More generally, a conventional probability sample survey can both be used to
assess the accuracy of the fast and cheap estimates, and they can also be used
to try to measure and relax some of the assumptions required by the faster,
cheaper strategy.

\clearpage

\hypertarget{references}{%
\section*{References}\label{references}}
\addcontentsline{toc}{section}{References}

\hypertarget{refs}{}
\leavevmode\hypertarget{ref-bauernschuster_surfing_2014}{}%
Bauernschuster, S., Falck, O., \& Woessmann, L. (2014). Surfing alone? The internet and social capital: Evidence from an unforeseeable technological mistake. \emph{Journal of Public Economics}, \emph{117}, 73--89.

\leavevmode\hypertarget{ref-bernard_counting_2010}{}%
Bernard, H. R., Hallett, T., Iovita, A., Johnsen, E. C., Lyerla, R., McCarty, C., et al. (2010). Counting hard-to-count populations: The network scale-up method for public health. \emph{Sexually Transmitted Infections}, \emph{86}(Suppl 2), ii11--ii15. \url{http://sti.bmj.com/content/86/Suppl_2/ii11.short}

\leavevmode\hypertarget{ref-bernard_estimating_1991}{}%
Bernard, H. R., Johnsen, E. C., Killworth, P. D., \& Robinson, S. (1991). Estimating the size of an average personal network and of an event subpopulation: Some empirical results. \emph{Social Science Research}, \emph{20}(2), 109--121.

\leavevmode\hypertarget{ref-billari_broadband_2018}{}%
Billari, F. C., Giuntella, O., \& Stella, L. (2018). Broadband internet, digital temptations, and sleep. \emph{Journal of Economic Behavior \& Organization}, \emph{153}, 58--76.

\leavevmode\hypertarget{ref-billari_does_2019}{}%
Billari, F. C., Giuntella, O., \& Stella, L. (2019). Does broadband Internet affect fertility? \emph{Population Studies}, \emph{0}(0), 1--20.

\leavevmode\hypertarget{ref-brass_methods_1975}{}%
Brass, W. (1975). Methods for estimating fertility and mortality from limited and defective data. \emph{Methods for estimating fertility and mortality from limited and defective data.} \url{http://www.cabdirect.org/abstracts/19762901082.html}

\leavevmode\hypertarget{ref-brewer_prostitution_2000}{}%
Brewer, D. D., Potterat, J. J., Garrett, S. B., Muth, S. Q., Roberts, J. M., Kasprzyk, D., et al. (2000). Prostitution and the sex discrepancy in reported number of sexual partners. \emph{Proceedings of the National Academy of Sciences}, \emph{97}(22), 12385--12388.

\leavevmode\hypertarget{ref-cesare_promises_2018}{}%
Cesare, N., Lee, H., McCormick, T., Spiro, E., \& Zagheni, E. (2018). Promises and Pitfalls of Using Digital Traces for Demographic Research. \emph{Demography}, \emph{55}(5), 1979--1999.

\leavevmode\hypertarget{ref-clarke_has_2006}{}%
Clarke, G. R. G., \& Wallsten, S. J. (2006). Has the Internet Increased Trade? Developed and Developing Country Evidence. \emph{Economic Inquiry}, \emph{44}(3), 465--484.

\leavevmode\hypertarget{ref-cohen_use_2011}{}%
Cohen, R. A., \& Adams, P. F. (2011). \emph{Use of the Internet for health information: United States, 2009}. US Department of Health and Human Services, Centers for Disease Control and Prevention, National Center for Health Statistics.

\leavevmode\hypertarget{ref-eckman_assessing_2014}{}%
Eckman, S., Kreuter, F., Kirchner, A., Jäckle, A., Tourangeau, R., \& Presser, S. (2014). Assessing the mechanisms of misreporting to filter questions in surveys. \emph{Public Opinion Quarterly}, \emph{78}(3), 721--733.

\leavevmode\hypertarget{ref-fatehkia_using_2018}{}%
Fatehkia, M., Kashyap, R., \& Weber, I. (2018). Using Facebook ad data to track the global digital gender gap. \emph{World Development}, \emph{107}, 189--209.

\leavevmode\hypertarget{ref-feehan_network_2015}{}%
Feehan, D. M. (2015). \emph{Network reporting methods} (PhD thesis). Princeton University. Retrieved from \url{http://gradworks.umi.com/37/29/3729745.html}

\leavevmode\hypertarget{ref-feehan_generalizing_2016}{}%
Feehan, D. M., \& Salganik, M. J. (2016a). Generalizing the Network Scale-Up Method: A New Estimator for the Size of Hidden Populations. \emph{Sociological Methodology}, \emph{46}(1), 153--186. \url{http://128.84.21.199/pdf/1404.4009.pdf}

\leavevmode\hypertarget{ref-feehan_surveybootstrap_2016}{}%
Feehan, D. M., \& Salganik, M. J. (2016b). \emph{Surveybootstrap: Tools for the Bootstrap with Survey Data}. \url{https://CRAN.R-project.org/package=surveybootstrap}

\leavevmode\hypertarget{ref-feehan_quantity_2016}{}%
Feehan, D. M., Umubyeyi, A., Mahy, M., Hladik, W., \& Salganik, M. J. (2016). Quantity Versus Quality: A Survey Experiment to Improve the Network Scale-up Method. \emph{American Journal of Epidemiology}, kwv287.

\leavevmode\hypertarget{ref-friemel_digital_2016}{}%
Friemel, T. N. (2016). The digital divide has grown old: Determinants of a digital divide among seniors. \emph{new media \& society}, \emph{18}(2), 313--331. \url{http://journals.sagepub.com/doi/abs/10.1177/1461444814538648}

\leavevmode\hypertarget{ref-goel_nonrepresentative_2015}{}%
Goel, S., Obeng, A., \& Rothschild, D. (2015). Non-representative surveys: Fast, cheap, and mostly accurate. In \emph{Working Paper}. \url{http://adamobeng.com/download/FastCheapAccurate.pdf}

\leavevmode\hypertarget{ref-greenwell_health_2018}{}%
Greenwell, F., \& Salentine, S. (2018). \emph{Health information system strengthening: Standards and best practices for data sources} (pp. 58--63). Chapel Hill, NC, USA: MEASURE Evaluation, University of North Carolina.

\leavevmode\hypertarget{ref-haight_revisiting_2014}{}%
Haight, M., Quan-Haase, A., \& Corbett, B. A. (2014). Revisiting the digital divide in Canada: The impact of demographic factors on access to the internet, level of online activity, and social networking site usage. \emph{Information, Communication \& Society}, \emph{17}(4), 503--519. \url{http://www.tandfonline.com/doi/abs/10.1080/1369118x.2014.891633}

\leavevmode\hypertarget{ref-hill_further_1977a}{}%
Hill, K., \& Trussell, J. (1977). Further developments in indirect mortality estimation. \emph{Population Studies}, \emph{31}(2), 313--334. \url{http://www.tandfonline.com/doi/abs/10.1080/00324728.1977.10410432}

\leavevmode\hypertarget{ref-hjort_arrival_2019}{}%
Hjort, J., \& Poulsen, J. (2019). The Arrival of Fast Internet and Employment in Africa. \emph{American Economic Review}, \emph{109}(3), 1032--1079.

\leavevmode\hypertarget{ref-icf_demographic_2004}{}%
ICF. (2004). \emph{Demographic and Health Surveys (various) {[}Datasets{]}.} Rockville, Maryland: ICF {[}Distributor{]}.

\leavevmode\hypertarget{ref-icf_what_2018}{}%
ICF. (2018). What we do: Survey process. \url{https://dhsprogram.com/What-We-Do/Survey-Process.cfm}

\leavevmode\hypertarget{ref-itu_percentage_2018}{}%
ITU. (2018). \emph{Percentage of individuals using the Internet}. Geneva: ITU (International Telecommunications Union). \url{https://www.itu.int/en/ITU-D/Statistics/Documents/statistics/2018/Individuals_Internet_2000-2017.xls}

\leavevmode\hypertarget{ref-kho_impact_2018}{}%
Kho, K., Lakdawala, L. K., \& Nakasone, E. (2018). Impact of Internet Access on Student Learning in Peruvian Schools, 51.

\leavevmode\hypertarget{ref-lavallee_indirect_2007}{}%
Lavallee, P. (2007). \emph{Indirect sampling}. New York: Springer-Verlag. \url{http://books.google.com/books?hl=en/\&lr=/\&id=o93cnlP9tMMC/\&oi=fnd/\&pg=PA1/\&dq=indirect+sampling+lavallee/\&ots=nhf1KvhIEk/\&sig=_7W13JSq39Iqe1WNclnr3HPk9ts}

\leavevmode\hypertarget{ref-lazer_computational_2009a}{}%
Lazer, D., Pentland, A., Adamic, L., Aral, S., Barabási, A.-L., Brewer, D., et al. (2009). Computational Social Science. \emph{Science}, \emph{323}(5915), 721--723.

\leavevmode\hypertarget{ref-lumley_analysis_2004}{}%
Lumley, T. (2004). Analysis of Complex Survey Samples. \emph{Journal of Statistical Software}, \emph{9}(1), 1--19.

\leavevmode\hypertarget{ref-lumley_complex_2011}{}%
Lumley, T. (2011). \emph{Complex Surveys: A Guide to Analysis Using R}. John Wiley \& Sons.

\leavevmode\hypertarget{ref-maltiel_estimating_2015}{}%
Maltiel, R., Raftery, A. E., McCormick, T. H., \& Baraff, A. J. (2015). Estimating population size using the network scale up method. \emph{Annals of Applied Statistics}, \emph{9}(3), 1247--1277.

\leavevmode\hypertarget{ref-manacorda_liberation_2016}{}%
Manacorda, M., \& Tesei, A. (2016). Liberation Technology: Mobile Phones and Political Mobilization in Africa. \emph{SSRN Electronic Journal}.

\leavevmode\hypertarget{ref-marsden_recent_2005}{}%
Marsden, P. V. (2005). Recent developments in network measurement. \emph{Models and methods in social network analysis}, \emph{8}, 30.

\leavevmode\hypertarget{ref-mossong_social_2008}{}%
Mossong, J., Hens, N., Jit, M., Beutels, P., Auranen, K., Mikolajczyk, R., et al. (2008). Social contacts and mixing patterns relevant to the spread of infectious diseases. \emph{PLoS Med}, \emph{5}(3), e74. \url{http://journals.plos.org/plosmedicine/article?id=10.1371/journal.pmed.0050074}

\leavevmode\hypertarget{ref-nic.br_survey_2016}{}%
NIC.br. (2016). \emph{Survey on the use of information and communication technologies in Brazilian households: ICT households 2015}. São Paulo: Comitê Gestor da Internet no Brasil. \url{https://cetic.br/media/docs/publicacoes/2/TIC_Dom_2015_LIVRO_ELETRONICO.pdf}

\leavevmode\hypertarget{ref-ofcom_adults_2016}{}%
Ofcom. (2016). \emph{Adults' media use and attitudes report}. London: Ofcom.

\leavevmode\hypertarget{ref-parsons_design_2014}{}%
Parsons, V. L., Moriarity, C. L., Jonas, K., Moore, T. F., Davis, K. E., \& Tompkins, L. (2014). Design and estimation for the national health interview survey, 2006-2015.

\leavevmode\hypertarget{ref-perrin_americans_2015}{}%
Perrin, A., \& Duggan, M. (2015). \emph{Americans' Internet Access: 2000-2015}. Pew Research Center.

\leavevmode\hypertarget{ref-pewresearchcenter_internet_2018}{}%
Pew Research Center. (2018). Internet/Broadband Fact Sheet. \url{http://www.pewinternet.org/fact-sheet/internet-broadband/}

\leavevmode\hypertarget{ref-rao_double_1968a}{}%
Rao, J. N. K., \& Pereira, N. P. (1968). On double ratio estimators. \emph{Sankhyā: The Indian Journal of Statistics, Series A (1961-2002)}, \emph{30}(1), 83--90.

\leavevmode\hypertarget{ref-rao_resampling_1988}{}%
Rao, J. N. K., \& Wu, C. F. J. (1988). Resampling inference with complex survey data. \emph{Journal of the American Statistical Association}, \emph{83}(401), 231--241.

\leavevmode\hypertarget{ref-rao_recent_1992}{}%
Rao, J., Wu, C., \& Yue, K. (1992). Some recent work on resampling methods for complex surveys. \emph{Survey Methodology}, \emph{18}(2), 209--217.

\leavevmode\hypertarget{ref-rojas_harnessing_2015}{}%
Rojas, G. (2015, January). Harnessing Technology to Streamline Data Collection. \url{https://blog.dhsprogram.com/harnessing-technology-streamline-data-collection/}

\leavevmode\hypertarget{ref-sarndal_model_2003}{}%
Sarndal, C. E., Swensson, B., \& Wretman, J. (2003). \emph{Model assisted survey sampling}. Springer Verlag. \url{http://books.google.com/books?hl=en/\&lr=/\&id=ufdONK3E1TcC/\&oi=fnd/\&pg=PR5/\&dq=sarndal+swensson+wretman+model+assisted/\&ots=7eZV4u7FOC/\&sig=tdK954DVTis0gvMz7r4SapBVnYg}

\leavevmode\hypertarget{ref-sarndal_estimation_2005}{}%
Särndal, C.-E., \& Lundström, S. (2005). \emph{Estimation in Surveys with Nonresponse}. John Wiley \& Sons.

\leavevmode\hypertarget{ref-sirken_household_1970}{}%
Sirken, M. G. (1970). Household surveys with multiplicity. \emph{Journal of the American Statistical Association}, \emph{65}(329), 257--266.

\leavevmode\hypertarget{ref-tourangeau_motivated_2015}{}%
Tourangeau, R., Kreuter, F., \& Eckman, S. (2015). Motivated misreporting: Shaping answers to reduce survey burden. \emph{Survey measurements. techniques, data quality and sources of error}, 24--41.

\leavevmode\hypertarget{ref-valliant_practical_2013}{}%
Valliant, R., Dever, J. A., \& Kreuter, F. (2013). \emph{Practical tools for designing and weighting survey samples}. New York: Springer.

\leavevmode\hypertarget{ref-vandeursen_digital_2014}{}%
Van Deursen, A. J., \& Van Dijk, J. A. (2014). The digital divide shifts to differences in usage. \emph{New media \& society}, \emph{16}(3), 507--526. \url{http://journals.sagepub.com/doi/abs/10.1177/1461444813487959}

\leavevmode\hypertarget{ref-vigdor_scaling_2014}{}%
Vigdor, J. L., Ladd, H. F., \& Martinez, E. (2014). Scaling the digital divide: Home computer technology and student achievement. \emph{Economic Inquiry}, \emph{52}(3), 1103--1119. \url{http://onlinelibrary.wiley.com/doi/10.1111/ecin.12089/full}

\leavevmode\hypertarget{ref-wasserman_social_1994}{}%
Wasserman, S., Faust, K., \& Urbana-Champaign), S. (. of I. W. (1994). \emph{Social Network Analysis: Methods and Applications}. Cambridge University Press.

\leavevmode\hypertarget{ref-wolter_introduction_2007}{}%
Wolter, K. (2007). \emph{Introduction to Variance Estimation} (2nd ed.). New York: Springer.

\leavevmode\hypertarget{ref-worldbank_world_2016}{}%
World Bank. (2016). \emph{World Development Report 2016: Digital Dividends}. Washington, D.C. \url{http://www.worldbank.org/en/publication/wdr2016}

\leavevmode\hypertarget{ref-zagheni_you_2012}{}%
Zagheni, E., \& Weber, I. (2012). You are where you e-mail: Using e-mail data to estimate international migration rates. In \emph{Proceedings of the 4th annual ACM web science conference} (pp. 348--351). ACM.

\clearpage
\appendix

\renewcommand{\thefigure}{S\arabic{figure}}
\renewcommand{\thetable}{S\arabic{table}}

\setcounter{figure}{0}

\hypertarget{supporting-information}{%
\section*{Supporting Information}\label{supporting-information}}
\addcontentsline{toc}{section}{Supporting Information}

\hypertarget{sec:estimator}{%
\section{Derivation of the estimators}\label{sec:estimator}}

\hypertarget{sampling-setup}{%
\subsubsection*{Sampling setup}\label{sampling-setup}}
\addcontentsline{toc}{subsubsection}{Sampling setup}

We assume a conventional probability sampling setup, following the
theory of design-based sampling; see Sarndal et al. (2003) for an
overview.
When we refer to an estimator as `consistent', we mean design-consistent
(also called Fisher consistent; Sarndal et al. (2003)). Similarly, `unbiased'
means design-unbiased.

Our frame population \(F\) -- the set of people who could potentially
be sampled -- is monthly active Facebook users in a given country\footnote{Throughout this paper, we use the term Facebook users to refer to monthly-active
  Facebook users.}.
The population whose size we are trying to estimate is \(H\), the number of
internet users in the country.
The goal is to use information about people on Facebook's reported offline
personal network connections in order to estimate the size of \(H\).

We assume that we obtain a \emph{probability sample} \(s\) from the frame population,
where we use the same definition of a probability sample as Sarndal et al. (2003).
To briefly review, we assume that the sample \(s\) is chosen from among the
members of the frame population \(F\) using a known random sampling method.
The probability that \(i \in F\) is included in the sample \(s\), called \(i\)'s
\emph{inclusion probability}, is written \(\pi_i\).
We require that \(\pi_i > 0\) for all \(i \in F\).
We call the \(w_i = \frac{1}{\pi_i}\) the \emph{expansion weight} for unit \(i \in F\).

Several of the estimators we study are ratio or compound ratio estimators.
The literature on design-based sampling has established that if
each component estimator is consistent and unbiased, then compound
ratio estimators are design-consistent but,
strictly speaking, compound ratio estimators are not unbiased.
Fortunately, a large literature has studied this problem and such estimators
are typically found to be very nearly unbiased, both in theory and in
practice\footnote{We do not expect the situations in which compound ratio estimators
  would be biased to be relevant to our study; the biggest concern is typically
  when the denominator of \(\widehat{R}\) is very small, which is not likely in
  our applications.}. Thus, we refer to these compound ratio estimators
as \emph{essentially unbiased}.
The following result formally establishes these important properties of
compound ratio estimators; which we will use these properties below.

~

\begin{Result}
\label{res:compound-ratio} Suppose that
\(\widehat{y}_1, \dots, \widehat{y}_n\) are estimators that are
consistent and unbiased for \(Y_1, \dots, Y_n\) respectively. Then the
compound ratio estimator

\begin{equation}
\widehat{R} = \frac{\widehat{y}_1 \dots \widehat{y}_k}{\widehat{y}_{k+1} \dots \widehat{y}_n}.
\end{equation}

is consistent and essentially unbiased for
\(R = (Y_1 \dots Y_{k}) / (Y_{k+1} \dots Y_n)\).
\end{Result}

\begin{proof}
See Rao and Pereira (1968), Wolter (2007) (pg. 233), and Feehan and
Salganik (2016a) for more details.
\end{proof}

We adhere to the notation used in previous papers about network scale-up and
network reporting (Feehan 2015; Feehan and Salganik 2016a; Feehan et al. 2016):

\begin{itemize}
\tightlist
\item
  \(y_{i, B}\) is the number of reported connections from person \(i\) to members of group \(B\)
\item
  \(y_{A,B} = \sum_{i \in A} y_{i,B}\) is the number of reported connections from members of group \(A\) to group \(B\)
\item
  \(d_{i,B}\) is the number of undirected connections in the social network between \(i\)
  and members of group \(B\)
\item
  \(d_{A,B} = \sum_{i \in A} d_{i, B}\) is the total number of undirected connections in the social network between members of group \(A\) and members of group \(B\)
\item
  \(v_{i,A}\) is the \emph{visibility} of \(i\) to group \(A\) -- i.e., the number of times that \(i\) would be reported if everyone in \(A\) was interviewed
\item
  \(v_{B, A} = \sum_{i \in B} v_{i,A}\) is the total visibility of members of group \(B\) to group \(A\)
\item
  \(\widehat{y} \rightarrow Y\) is shorthand for `\(\widehat{y}\) is a consistent and unbiased estimator for \(Y\)'
\item
  \(\widehat{y} \rightsquigarrow Y\) is shorthand for `\(\widehat{y}\) is a consistent and essentially unbiased estimator for \(Y\)'
\item
  \(y^+_{F,H}\) is the number of reported connections from \(F\) to \(H\) that actually lead to \(H\). If \(y^+_{F,H} = y_{F,H}\) then we say that there are \emph{no false positive reports}
\item
  \(N_A\) is the size of set \(A\) (i.e., the number of people in \(A\))
\end{itemize}

\hypertarget{aggregate-reporting-framework}{%
\subsubsection*{Aggregate reporting framework}\label{aggregate-reporting-framework}}
\addcontentsline{toc}{subsubsection}{Aggregate reporting framework}

We develop an estimator using the network reporting framework, an approach that
builds upon insights from several different streams of previous research on
sampling
(Bernard et al. 1991; Feehan 2015; Feehan and Salganik 2016a; Lavallee 2007; Sirken 1970).
Feehan (2015) shows that researchers can develop estimators based on network
reports using either an individual or an aggregate multiplicity approach.
Since we do not collect information at the level of detail required by
individual multiplicity estimation, we adopt an aggregate multiplicity approach
in this study.
This aggregate multiplicity approach is similar to the network scale-up method
(Bernard et al. 2010, 1991; Feehan and Salganik 2016a; Maltiel et al. 2015).

~

\begin{Result}
\label{res:agg-mult-id} Suppose that a census of the frame population
\(F\) is interviewed and asked to report about their connections to a
group \(Z\). Call the total number of reported connections \(y_{F,Z}\)
and suppose \(y_{F,Z} > 0\). Further, suppose that there are no false
positive reports, so that \(y_{F,Z} = y^+_{F,Z}\). Finally, suppose that
\(\bar{v}_{Z,F}\) is the average visibility of members of \(Z\); that
is, \(\bar{v}_{Z,F}\) is the average number of times that a member of
\(Z\) is reported by someone in \(F\). Then \begin{equation}
N_H = \frac{y_{F,H}}{\bar{v}_{H,F}}.
\label{eq:ap-aggmult-id}
\end{equation}
\end{Result}

\begin{proof}
See Feehan (2015) and Feehan and Salganik (2016a).
\end{proof}

To see the intuition behind the aggregate multiplicity approach from
Result \ref{res:agg-mult-id}, suppose we
conducted a census of the frame population, asking every frame population
member to tell us how many members of her personal network were online.
Simply adding up the number of reported connections to internet users would
produce a number that is larger than the number of internet users because each internet user can be reported more than once.
Thus, in order to adjust for this over-counting, aggregate multiplicity estimators divide
an estimate for the total number of reports by an estimate of hidden population
members' \emph{visibility}.
The visibility is the number of times an average member of the hidden
population would be reported if everyone on the frame population responded to
the survey.
In this study, the visibility is the number of times that the average internet
user in a given country would be reported as an internet user, if everyone on Facebook in the
country responded to the survey. Dividing the estimated total number of reported connections
to people on the internet by the estimated visibility adjusts for the over-counting that would
occur if the reports were used to directly estimate the number of internet users.

Given the aggregate multiplicity identity, our basic approach is to develop data
collection strategies and statistical estimators that enable us to estimate the
numerator and denominator of the identity in Eq. \ref{eq:ap-aggmult-id}.
In the remainder of this Appendix, we develop necessary technical results to use
the identity in Eq. \ref{eq:ap-aggmult-id} to estimate the number of internet users in a given
country.

\hypertarget{estimates-about-detailed-alters}{%
\subsubsection*{Estimates about detailed alters}\label{estimates-about-detailed-alters}}
\addcontentsline{toc}{subsubsection}{Estimates about detailed alters}

Result \ref{res:yhat-f-z} formalizes a situation where respondents are sampled
and then asked about a sample of their network members.
Result \ref{res:yhat-f-z} is stated in terms of an arbitrary dichotomous trait \(z\) that
respondents report about their personal network members; for example, \(z\)
could be Facebook usage, internet usage, gender, or membership in an age group.

~

\begin{Result}
\label{res:yhat-f-z} Suppose we have a sample \(s\) taken from the frame
population using a probability sampling design. Call the expansion
weights given by the sampling design \(w_i\) for each \(i \in s\).
Further, suppose that for each \(i \in s\), we obtain information from a
simple random subsample \(s_{i}\) of size \(r_i\) from the \(d_i\)
people in \(i\)'s personal network. Let \(z_{ij}\) be an indicator
variable for whether or not \(i\) reports that \(j\) has trait \(Z\),
and let \(z_i = \sum_{j \in s_i} z_{ij}\) be the total number of
detailed alters respondent \(i\) reports having trait \(Z\). Then the
estimator

\begin{equation}
\widehat{y}_{F,Z} = \sum_{i \in s} w_i \frac{d_i}{r_i} z_i
\end{equation}

is consistent and unbiased for \(y_{F,Z}\), the total number of reported
connections to people with trait \(Z\) in a census of the frame
population in which respondents report about everyone in their networks.
\end{Result}

\begin{proof}
First, we note that we can consider this to be a multi-stage sample,
where the first stage(s) lead to selection of the respondent and the
final stage is the subsampling of detailed alters within each
respondent's network. Since the final stage is a simple random sample of
\(r_i\) out of \(d_i\) network members, the design weight for the final
stage is \(\frac{d_i}{r_i}\) for each detailed alter. In order to show
that the estimator is unbiased, we take expectations with respect to the
multi-stage sampling design:

\begin{equation}
\begin{aligned}
 \mathbb{E} [\widehat{y}_{F,Z}] 
&=  \mathbb{E} _I[\sum_{i \in s} w_i  \mathbb{E} _{i}[\frac{d_i}{r_i} z_i | s ]]\\
&= \sum_{i \in F} \pi_i w_i  \mathbb{E} _{i}[\frac{d_i}{r_i} z_i | s]]\\
&= \sum_{i \in F} \pi_i w_i \left( \sum_{j \sim i} \pi^i_j \frac{d_i}{r_i} z_{ij}\right),
\end{aligned}
\end{equation}

where the outer expectation \( \mathbb{E} _I[\cdot]\) is taken with
respect to the sampling of respondents and the inner expectation
\( \mathbb{E} _i[\cdot | s]\) is taken with respect to the sampling of
detailed alters within each sampled respondent; \(j \sim i\) indexes
over all of the network members \(j\) that \(i\) could potentially
report about; and we have written \(\pi_i\) for the inclusion
probability of respondent \(i\) under the sampling design, and
\(\pi^i_j\) for the inclusion probability of respondent \(i\)'s \(j\)th
network member under the subsampling design.

By definition, \(w_i = \frac{1}{\pi_i}\) and
\(\pi^i_j = \frac{r_i}{d_i}\). Thus, continuing from above, we have

\begin{equation}
\begin{aligned}
 \mathbb{E} [\widehat{y}_{F,Z}] 
&= \sum_{i \in F} \pi_i w_i \left( \sum_{j \sim i} \pi^i_j \frac{d_i}{r_i} z_{ij}\right)\\
&= \sum_{i \in F} \left( \sum_{j \sim i} \pi^i_j \frac{d_i}{r_i} z_{ij}\right)\\
&= \sum_{i \in F} y_{i, Z}\\
&= y_{F, Z}. 
\end{aligned}
\end{equation}

So we have shown that the estimator is unbiased for \(y_{F,Z}\).

Finally, in a census of the frame population where every respondent
reports about all of her network members, \(s = F\), \(\pi_i = 1\),
\(\pi^i_j = 1\), \(z_i = y_{i,Z}\), and \(r_i = d_i\) for all \(i\) and
\(j\). Thus

\begin{equation}
\widehat{y}_{F,Z} 
= \sum_{i \in s} w_i \frac{d_i}{r_i} z_i
= \sum_{i \in F} y_{i,Z}
= y_{F,Z}
\end{equation}

So the estimator is design-consistent.
\end{proof}

\begin{Corollary}
\label{res:ybar-f-z} Under the conditions of Result \ref{res:yhat-f-z},
the estimator

\begin{equation}
\widehat{\bar{y}}_{F,Z} = \frac{\sum_{i \in s} w_i \frac{d_i}{r_i} z_i}{\sum_{i \in s} w_i}
\end{equation}

is consistent and essentially unbiased for \(\bar{y}_{F, Z}\).
\end{Corollary}

\begin{proof}
By Result \ref{res:yhat-f-z}, the numerator is consistent and unbiased
for \(y_{F,Z}\), and the denominator is a sample-based estimate for the
size of the frame population, \(\widehat{N}_F = \sum_{i \in s} w_i\).
Thus, this is a Hajek-type estimator. See (Sarndal et al. 2003) for a
proof that Hajek estimators are consistent and essentially unbiased.
\end{proof}

Note that
Result \ref{res:yhat-f-z} implies that Eq. \ref{eq:estimator-numerator} is
consistent and unbiased for \(y_{F,H}\) and
Corollary \ref{res:yhat-f-z} implies that Eq. \ref{eq:degree-estimator} is
consistent and unbiased for \(\bar{y}_{F,F}\).

\hypertarget{assembling-the-estimator}{%
\subsubsection*{Assembling the estimator}\label{assembling-the-estimator}}
\addcontentsline{toc}{subsubsection}{Assembling the estimator}

The next estimator, Result \ref{res:nh}, shows that if we can estimate
the total reported connections from frame population members to internet users,
and if we can estimate the average visibility of internet users to frame
population members, then we can estimate the number of internet users.

~

\begin{Result}
\label{res:nh} Suppose that the \(\widehat{y}_{F,H}\) is a consistent
and unbiased estimator for \(y_{F,H}\) and that
\(\widehat{\bar{y}}_{F,F}\) is a consistent and essentially unbiased
estimator for \(\bar{y}_{F,F}\). Further, suppose that reports are
accurate in aggregate, so that \(y_{F,H} = d_{F,H}\) and
\(y_{F,F} = d_{F,F}\). Finally, suppose that

\begin{equation}
\begin{aligned}
\bar{d}_{H,F} &= \bar{d}_{F,F}.
\end{aligned}
\label{eq:ap-equal-tie-fbnofb}
\end{equation}

Then the estimator

\begin{equation}
\widehat{N}_H = \frac{\widehat{y}_{F,H}}{\widehat{\bar{y}}_{F,F}}
\end{equation}

is consistent and essentially unbiased for \(N_H\).
\end{Result}

\begin{proof}
Since \(\widehat{y}_{F,H} \rightarrow y_{F,H}\) and
\(\widehat{\bar{y}}_{F,F} \rightsquigarrow \bar{y}_{F,F}\), Result
\ref{res:compound-ratio} shows that
\(\widehat{N}_H = \frac{\widehat{y}_{F,H}}{\widehat{\bar{y}}_{F,F}} \rightsquigarrow \frac{y_{F,H}}{\bar{y}_{F,F}}\).
It remains to show that \(\frac{y_{F,H}}{\bar{y}_{F,F}}\) is equal to
\(N_H\). By the condition that reports are accurate in aggregate,
\(y_{F,H} = d_{F,H}\) and \(y_{F,F} = d_{F,F}\). Thus,

\begin{equation}
\frac{y_{F,H}}{\bar{y}_{F,F}} = \frac{d_{F,H}}{\bar{d}_{F,F}}.
\end{equation}

Next, using the condition that \(\bar{d}_{F,F} = \bar{d}_{H,F}\), we
have

\begin{equation}
\frac{d_{F,H}}{\bar{d}_{F,F}}
= \frac{d_{F,H}}{\bar{d}_{H,F}} = N_H \frac{d_{F,H}}{d_{H,F}} = N_H,
\end{equation}

where the last step follows from the fact that we are assuming a
symmetric type of network tie, meaning that the number of connections
from \(F\) to \(H\) must be equal to the number of connections from
\(H\) to \(F\).
\end{proof}

Result \ref{res:nh} relies upon
the condition that \(\bar{d}_{H,F} = \bar{d}_{F,F}\) (Eq. \ref{eq:ap-equal-tie-fbnofb}), which requires that
two quantities be equal: (1) the rate at which someone who is on the internet
shares a meal with someone who is on Facebook (\(\bar{d}_{H,F}\));
and, (2) the rate at which someone who is on Facebook shares a meal with
someone who is also on Facebook (\(\bar{d}_{F,F}\)).
This assumption could be violated if, for example, people frequently
organize sharing a meal together using Facebook (without inviting other
people).

To further understand the condition in Eq. \ref{eq:ap-equal-tie-fbnofb},
note that since \(F \subset H\) (i.e., everyone on Facebook is also on the
Internet), it follows that

\begin{align}
\bar{d}_{H,F} &= p_{F|H} \bar{d}_{F,F} + (1-p_{F|H}) \bar{d}_{H-F,F}
\end{align}

\noindent where \(p_{F|H} = \frac{N_F}{N_H}\) is the prevalence of \(F\) among \(H\), i.e.,
the fraction of people on the internet that is also on Facebook.
Therefore, when the condition in Eq. \ref{eq:ap-equal-tie-fbnofb} holds,
then it is also the case that

\begin{align}
\bar{d}_{F,F} = \bar{d}_{H-F, F}.
\end{align}

Appendix \ref{sec:sensitivity-analysis} introduces a sensitivity framework that researchers
can use to assess how sensitive size estimates are to this condition,
and Appendix \ref{sec:simple-model} introduces simple models that motivate this condition.

\hypertarget{sec:internal-consistency}{%
\section{Internal consistency checks}\label{sec:internal-consistency}}

The internal consistency checks start from an identity that relates
two quantities:
(1) \(d_{F_{-\alpha},F_\alpha}\), the population-level number of connections from everyone who is in \(F\) but not
group \(\alpha\) to everyone who is in \(F\) and in group \(\alpha\);
and (2), \(d_{F_{\alpha}, F_{-\alpha}}\) -- the population-level number of connections from everyone who is
in \(F\) and group \(\alpha\) to everyone who is in \(F\) but not in group \(\alpha\).
Since the networks we ask respondents to report about are symmetric, these two
quantities are identical; however, they can be estimated independently from the data
we collected: the first quantity can be estimated only from respondents who are
not in group \(\alpha\), and the second quantity can be estimated only from
respondents who are in group \(\alpha\).

In order to assess how internally consistent reporting is, we can directly
compute a survey-based estimates for the discrepancy

\begin{equation}
\Delta^0_\alpha = \widehat{d}_{F_{-\alpha},F_\alpha} - \widehat{d}_{F_{\alpha}, F_{-\alpha}}.
\end{equation}

\noindent The closer this quantity is to 0, the more internally consistent reports about group
\(\alpha\) are. However, \(\Delta^0_\alpha\) is influenced by the size of the group \(\alpha\), which makes it
challenging to plot internal consistency checks for several different
groups in the same place (e.g.~Fig. \ref{fig:ic}). Thus, we propose rescaling the IC checks for different
groups to put them on a more similar scale. Specifically, we scale \(\Delta^0_\alpha\) by a factor \(K\) given by

\begin{equation}
\begin{aligned}
K = \frac{N_F}{N_{F_{-\alpha}} N_{F_\alpha}},
\end{aligned}
\label{eq:ic-k}
\end{equation}

\noindent where \(N_{F_{-\alpha}}\) is the number of people in the frame population not in group \(\alpha\)
and \(N_{F_\alpha}\) is the number of people in the frame population who are in group \(\alpha\).
The factor \(K\) is motivated by starting from the identity \(d_{F_{-\alpha},F_\alpha} = d_{F_{\alpha}, F_{-\alpha}}\),
and multiplying both sides by \(\frac{1}{N_{F_{-\alpha}} N_{F_\alpha}}\). The result is an expression that shows that
\(\bar{d}_{F_{-\alpha}, F_{\alpha}}/N_{F_\alpha} = \bar{d}_{F_{\alpha}, F_{-\alpha}} / N_{F_{-\alpha}}\). In words,
this new expression equates (1) the proportion of \(F_\alpha\) that the average person in \(F_{-\alpha}\) is connected to;
and (2) the proportion of \({F_{-\alpha}}\) that the average person in \(F_{\alpha}\) is connected to.
Finally, we multiply the new identity by \(N_F\) to help compare countries of different sizes.

Note that this rescaling does not affect whether or not the confidence intervals
for the IC checks includes 0; instead, it controls for the relative size of
group \(\alpha\). It makes internal consistency checks across different groups easier to compare with one another.

Using the example of the conversational contact reports, the final discrepancy measure is defined to be

\begin{equation}
\begin{aligned}
\Delta_\alpha^{\text{cc}} = K \left[ \widehat{d}_{F_{-\alpha},F_\alpha} - \widehat{d}_{F_{\alpha}, F_{-\alpha}} \right].
\end{aligned}
\label{eq:ic-bigdelta}
\end{equation}

\noindent Eq. \ref{eq:ic-bigdelta} can be computed for each bootstrap resample; the
distribution of \(\Delta_\alpha^{\text{cc}}\) across bootstrap resamples is then
an estimate for the sampling distribution of the discrepancy measure.

\hypertarget{sec:sensitivity-analysis}{%
\section{Sensitivity framework}\label{sec:sensitivity-analysis}}

In this Appendix, we describe a framework that can be used to
assess the sensitivity of the estimated number of people who use
the internet to the various conditions that the results in
Appendix \ref{sec:estimator} rely upon.

In order to develop the sensitivity framework, we adapt previous
work on network scale-up and other network reporting methods
(Feehan 2015; Feehan and Salganik 2016a).
We start by introducing three quantities, called \emph{adjustment factors}:

\begin{equation}
\eta_{H} = \frac{\substack{\text{avg \# reported connections from} \\ \text{F to H that actually
lead to H}}}{\substack{\text{avg \# reported connections} \\ \text{from F to H}}}
= \frac{y^+_{F,H}}{y_{F,H}},
\label{eq:defn-eta-h}
\end{equation}

and

\begin{equation}
\eta_{F} = \frac{\substack{\text{avg \# reported connections from} \\ \text{F to F that actually
lead to F}}}{\substack{\text{avg \# reported connections} \\ \text{from F to F}}}
= \frac{y^+_{F,F}}{y_{F,F}},
\label{eq:defn-eta-f}
\end{equation}

and

\begin{equation}
\nu = \frac{\text{avg \# in-reports to H from F}}{\text{avg \# in-reports to F from F}}
= \frac{\bar{v}_{H,F}}{\bar{v}_{F,F}}.
\label{eq:defn-nu}
\end{equation}

Each of these new parameters is equal to 1 under ideal conditions, when
the requirements of the results in Appendix \ref{sec:estimator} are satisfied.
In general, \(\nu\) can take on any value from \(0\) to \(\infty\),
while \(\eta_F\) and \(\eta_H\) can take on any value from \(0\) to \(1\).

The first sensitivity result reveals how estimated numbers of internet users
will be affected if one or more of the three adjustment factors is not equal
to 1.

~

\begin{Result}
\label{res:sens-yhat-f-z} Suppose that the sampling conditions for
Result \ref{res:yhat-f-z} hold, but that the reporting and network
structure conditions do not. That is, suppose we have a sample \(s\)
taken from the frame population using a probability sampling design.
Call the expansion weights given by the sampling design \(w_i\) for each
\(i \in s\). Further, suppose that for each \(i \in s\), we obtain
information from a simple random subsample \(s_{i}\) of \(r_i\) out of
the \(d_i\) people in \(i\)'s personal network.

Now suppose that \(\widehat{y}_{F,H}\) is consistent and unbiased for
\(y_{F,H}\) and that \(\widehat{\bar{y}}_{F,F}\) is consistent and
unbiased for \(\bar{y}_{F,F}\), but that \(\eta_{F,H} \neq 1\),
\(\eta_{F,F} \neq 1\), and \(\nu \neq 1\); that is, assume that the
remaining conditions in Result \ref{res:nh} do not hold. Then the
estimator

\begin{equation}
\widehat{N}_H = \frac{\widehat{y}_{F,H}}{\widehat{\bar{y}}_{F,F}}
\end{equation}

is consistent and unbiased for \((\frac{\eta_F}{\eta_H} \nu ) N_H\).
\end{Result}

\begin{proof}
The proof follows along the lines of Feehan and Salganik (2016a).
Briefly,

\begin{equation}
\begin{aligned}
\widehat{N}_H &= \frac{\widehat{y}_{F,H}}{\widehat{\bar{y}}_{F,F}}
\rightsquigarrow \frac{y_{F,H}}{\bar{y}_{F,F}}\\
\end{aligned}
\end{equation}

by the sampling conditions. Next, we wish to use the adjustment factors
to relate the estimand to \(N_H\):

\begin{equation}
\begin{aligned}
\frac{y_{F,H}}{\bar{y}_{F,F}} 
&= \frac{\eta_F}{\eta_H}~ \frac{y^+_{F,H}}{\bar{y}^+_{F,F}}\\
&= \frac{\eta_F}{\eta_H}~ \frac{v_{H,F}}{\bar{v}_{F,F}}\\
&= \frac{\eta_F}{\eta_H}~ \frac{\bar{v}_{H,F}}{\bar{v}_{F,F}}~N_H\\
&= \frac{\eta_F}{\eta_H}~\nu~N_H.
\end{aligned}
\end{equation}

Thus, we conclude that

\begin{equation}
\widehat{N}_H \rightsquigarrow \frac{\eta_F}{\eta_H}\nu N_H.
\end{equation}
\end{proof}

\begin{Corollary}
\label{res:sens-yhat-f-z-bias} Under the conditions listed in Result
\ref{res:sens-yhat-f-z},

\begin{equation}
\text{Bias}[\widehat{N}_H] 
=  \mathbb{E} [\widehat{N}_H] - N_H
= N_H (\frac{\eta_F}{\eta_H}\nu - 1).
\end{equation}
\end{Corollary}

Now we show how problems with the sampling weights can affect estimates;
this will be helpful in understanding what impact non simple random
subsampling of detailed alters would have.

First, we must define \emph{imperfect sampling weights}. We follow Feehan and Salganik (2016a)
and repeat the definition here for convenience:

\textbf{Imperfect sampling weights.} Suppose a researcher obtains a probability sample
\(s\) from the frame population \(F\). Let \(I_i\) be the random variable that
assumes the value 1 when unit \(i \in F\) is included in the sample \(s\),
and 0 otherwise. Let \(\pi_i =  \mathbb{E} [I_i]\) be the true probability of inclusion
for unit \(i \in F\), and let \(w_i = \frac{1}{\pi_i}\) be the corresponding
design weight for unit \(i\). We say that researchers have \emph{imperfect sampling weights}
when researchers use imperfect estimates of the inclusion probabilities \(\pi_i^\prime\)
and the corresponding design weights \(w_i^\prime = \frac{1}{\pi_i^\prime}\).
Note that we assume that both the true and the imperfect weights satisfy
\(\pi_i > 0\) and \(\pi_i^\prime > 0\) for all \(i\).

~\\
\hspace*{0.333em}

\begin{Result}
\label{res:imperfect-w} Suppose researchers have obtained a probability
sample \(s\), but that they have imperfect sampling weights. Call the
imperfect sampling weights \(w_i^\prime = \frac{1}{\pi_i^\prime}\), call
the true weights \(w_i = \frac{1}{\pi_i}\), and define
\(\epsilon_i = \frac{w_i^\prime}{w_i} = \frac{\pi_i}{\pi_i^\prime}\).
Then

\begin{equation}
\text{Bias}[\widehat{y}^\prime_{F,Z}]
= N_F \left[ \bar{y}_{F,Z} (\bar{\epsilon} - 1) + \text{cov}_F(y_{i,Z}, \epsilon_i)\right],
\end{equation}

where \(\bar{\epsilon} = \frac{1}{N_F} \sum_{i \in F} \epsilon_i\) and
\(\text{cov}_F(\cdot, \cdot)\) is the finite population unit covariance
in the frame population \(F\).
\end{Result}

\begin{proof}
See Result D.2 in Feehan and Salganik (2016a).
\end{proof}

Result \ref{res:imperfect-w} will be useful to us because we can
use it to understand situations where respondents' reports about the
detailed alters are different from simple random sampling.
In order to isolate the impact of such a difference, we assume in
Result \ref{res:imperfect-w} that the expansion weights for
respondent inclusion are accurate.

We also state the following fact, which will be useful in the
subsequent derivation.
~

\begin{Fact}
\label{res:sumprod-cov} \begin{equation}
\sum_{i \in A} a_i b_i = N_A[\bar{a}\bar{b} + \text{cov}_A(a_i,b_i)]
\end{equation}
\end{Fact}

\begin{Result}
\label{res:imp-detailed-alters} Suppose that respondents do not report
about the detailed alters by picking \(r_i\) out of \(d_i\) of them
uniformly at random, so that the estimator for \(\widehat{y}_{F,Z}\) in
Result \ref{res:yhat-f-z} uses imperfect weights
\(l_{ij}^\prime = \frac{d_i}{r_i}\) for the final-stage subsampling of
detailed alters, while the true weight for each of respondent \(i\)'s
detailed alters \(j\) is given by \(l_{ij}\). Let
\(\epsilon_i = \frac{l_i^\prime}{l_i}\). Suppose also that the expansion
weights \(w_i\) for the inclusion of respondents in the sample are
accurate. Then the bias of \(\widehat{y}_{F,Z}^\prime\) is given by

\begin{equation}
\text{Bias}[\widehat{y}_{F,Z}^\prime] =
\sum_{i \in F} \sum_{j \sim i} z_{ij} (\epsilon_{ij} - 1).
\end{equation}
\end{Result}

\begin{proof}
\begin{equation}
\begin{aligned}
 \mathbb{E} [\widehat{y}_{F,Z}^\prime] &=
 \mathbb{E} \left[\sum_{i \in s} w_i \times  \mathbb{E} _i [ \sum_{j \in s_i} l_{ij}^\prime z_{ij} | s] \right]\\
&= \sum_{i \in F} w_i  \mathbb{E} [I_i] \times \sum_{j \sim i}  \mathbb{E} [I_{ij} | s] l_{ij}^\prime z_{ij} \\
&= \sum_{i \in F} \sum_{j \sim i} \frac{l_{ij}^\prime}{l_{ij}} z_{ij} \\
&= \sum_{i \in F} \sum_{j \sim i } \epsilon_{ij} z_{ij},
\end{aligned}
\end{equation} where \(j \sim i\) indexes the people \(j\) that are
reported in respondent \(i\)'s network. Thus, the bias is
\begin{equation}
\begin{aligned}
\text{Bias}(\widehat{y}_{F,Z}^\prime) &=  \mathbb{E} [\widehat{y}_{F,Z}^\prime] - y_{F,Z}\\
&= \sum_{i \in F} \sum_{j \sim i} \epsilon_{ij} z_{ij} - \sum_{i \in F} \sum_{j \sim i} z_{ij}\\
&= \sum_{i \in F} \sum_{j \sim i} z_{ij} (\epsilon_{ij} - 1).
\end{aligned}
\end{equation}
\end{proof}

To understand Result \ref{res:imp-detailed-alters} better,
we manipulate the expression for \(\text{Bias}[\widehat{y}_{F,Z}^\prime]\) with the
aim of producing a more interpretable expression:

\begin{equation}
\begin{aligned}
\text{Bias}(\widehat{y}_{F,Z}^\prime) 
&= \sum_{i \in F} \sum_{j \sim i} z_{ij} (\epsilon_{ij} - 1)\\
&= \sum_{i \in F} y_i [\bar{z}_i (\bar{\epsilon}_i - 1) + \text{cov}_{j \sim i}(z_{ij}, \epsilon_{ij}-1)]\\
&= \sum_{i \in F} y_i \bar{z}_i \bar{\epsilon}_i 
   - \sum_{i \in F} y_i \bar{z}_i
   + \sum_{i \in F} y_i \sigma_i\\ 
&= \sum_{i \in F} z_i \bar{\epsilon}_i 
   + \sum_{i \in F} y_i \sigma_i
   - \sum_{i \in F} z_i,
\end{aligned}
\end{equation}

where \(j \sim i\) indexes the people \(j\) that are reported in respondent \(i\)'s network;
\(y_i = y_{i,U} = \sum_{j \sim i} 1\) is the total number of people \(i\) would report
about if there was no subsampling;
\(z_i = \sum_{i \sim j} z_{ij}\) is the total number of people \(i\) would report as members of \(Z\)
if there was no subsampling;
\(\bar{z}_i = y_{i}^{-1} \sum_{i \sim j} z_{ij}\) is the average \(z_{ij}\)
among respondent \(i\)'s reported network members;
\(\bar{z}=N_F^{-1}\sum_{i \in F} \bar{z}_i\) is the average \(\bar{z}_i\) across people in the frame;
\(\bar{\epsilon}_i = y_{i}^{-1} \sum_{i \sim j} \epsilon_{ij}\) is the average \(\epsilon_{ij}\)
among respondent \(i\)'s reported network members;
\(\bar{\epsilon}=N_F^{-1}\sum_{i \in F} \bar{\epsilon}_i\) is the average \(\bar{\epsilon}_i\) across people in the frame;
and \(\sigma_i = \text{cov}_{i \sim j}(z_{ij}, \epsilon_{ij})\) is the covariance between the \(\epsilon_{ij}\) and \(z_{ij}\) among respondent \(i\)'s reported network members.

Finally, we use Fact \ref{res:sumprod-cov} twice--once within respondent and
once between respondents:

\begin{equation}
\begin{aligned}
\sum_{i \in F} z_i \bar{\epsilon}_i 
   &+ \sum_{i \in F} y_i \sigma_i
   - \sum_{i \in F} z_i 
   \\
&= N_F[\bar{z}\bar{\epsilon} + \text{cov}_F(z_i, \bar{\epsilon}_i)]
   + N_F[\bar{y}_{F,U} \bar{\sigma} + \text{cov}_F(y_i, \sigma_i)]
   - y_{F,Z}\\
&= 
y_{F,Z} \left[
\underbrace{
(\bar{\epsilon}-1)}_{\substack{\text{aggregate} \\ \text{error} \\ \text{in weights}}}
+ 
\underbrace{
\frac{\bar{\sigma}~\bar{y}_{F,U} + \text{cov}_F(y_i, \sigma_i) }{\bar{y}_{F,Z}}
}_{\substack{\text{relationship between} \\ \text{personal network} \\ \text{size and} \\ \text{weight errors}}}f 
+
\underbrace{
\frac{\text{cov}_F(z_i, \bar{\epsilon}_i)}{\bar{y}_{F,Z}}\
}_{\substack{\text{relationship} \\ \text{between} \\ \text{weight errors} \\ \text{and alters'} \\ \text{internet use}}}
\right]\\
\end{aligned}
\label{eq:detailed-alter-decomp}
\end{equation}

Thus, Eq. \ref{eq:detailed-alter-decomp} shows that when respondents do not choose
detailed alters uniformly at random, the resulting bias can be decomposed into
three terms: one term related to aggregate errors in the weights; one term
that captures the relationship between personal network size and weight errors;
and one term that captures the relationship between weight errors and alters' internet use.

\hypertarget{sec:survey-methods}{%
\section{Additional details about survey design}\label{sec:survey-methods}}

In this appendix, we provide additional details about the design of our survey.
We hope that these details will help researchers who wish to use our design as a
starting point for future research.

\textbf{Weighting and post-stratification}

The estimator we develop in the main text is based on having a probability sample of the frame population.
We obtained a probability sample using Facebook's internal survey sampling mechanism.
However, like any real-world sampling approach, the actual set of people we interview is the result of a two-step process: first, we randomly chose people to be included in the sample;
and, second, some of those people agreed to participate in the study.
If the people who agreed to participate were systematically different from people who did not, that could affect
the accuracy of our inferences (just like any survey). Therefore, we adjust the sampling weights using post-stratification to improve the representativeness of our sample\footnote{Note that our goal is not to use post-stratification to account for differences between the people who use Facebook and other people; instead, the goal is to ensure that our sample is as representative as possible of people who use Facebook.}.
Post-stratification is commonly used in sample surveys with the goal of helping to reduce the bias and variance of
estimates;
readers who are interested in learning more about post-stratification can
consult standard survey research texts (e.g., Särndal and Lundström 2005; Lumley 2011; Valliant et al. 2013).
We provide a conceptual overview of how we used post-stratification here.

Our \emph{design weights} come from the sampling design, which we treat as a simple random sample without replacement. Under this sampling design, the probability of inclusion, i.e., the probability that a given person \(i\) who uses Facebook is included in the sample, is

\[
\pi^0_i = \frac{n}{N_F},\\
\]

where \(n\) is the sample size in \(i\)'s country and \(N_F\) is the size of the frame population, i.e., the number of active Facebook users in \(i\)'s country.
The design weight is the reciprocal of the probability of inclusion, \(w^0_i = \frac{1}{\pi^0_i}\).

In any particular sample \(s \subset F\), the set of respondents may be different from the
overall population of Facebook users. To account for this fact,
post-stratification adjusts these design weights by using known counts of active Facebook users by age and sex.
For a given survey respondent \(i\) who is in age-sex group \(\alpha\), the post-stratified weight is given by

\[
w_i = w^0_i~K_{\alpha},
\]
where \(K_{\alpha}\) is a factor given by

\[
K_\alpha = \frac{N_{F_\alpha}}{\sum_{i \in s_\alpha} w^0_i}.
\]

\(\sum_{i \in s_\alpha} w_i^0\) is the sum of the design weights among
respondents who are in age-sex group \(\alpha\),
and \(N_{F_\alpha}\) is the number of people who actively use Facebook in age-sex group \(\alpha\)
(or, more generally, the size of the frame population in group \(\alpha\)).
The intuition is that the denominator is the
design-weighted estimate of the number of frame population members in group \(\alpha\),
while \(N_{F_\alpha}\) is the actual size of group \(\alpha\), which is known.
\(K_\alpha\) thus adjusts the design weights so that they agree with the
known total sizes.

Recall that, in our study, we estimate sampling variation using the rescaled bootstrap
(Rao and Wu 1988; Rao et al. 1992).
We combine post-stratification and the bootstrap by post-stratifying
the set of weights that results from each bootstrap resample by age and sex\footnote{Specifically, used the \texttt{calibrate} function in the \texttt{survey} R package (Lumley 2004, 2011).}.
Thus, after post-stratification, for each bootstrap resample the sum of survey
weights will conform to known frame population totals.
For example, the sum of weights among female respondents in Brazil will equal the number of
females who are active on Facebook in Brazil.

In our study, post-stratifying the weights did not make a big difference in our estimates.
To illustrate this fact, Figure \ref{fig:estimates-nocalib} compares estimated internet adoption using the design
weights to estimated internet adoption using the post-stratified weights; the estimates are
very similar to one another.

However, in other studies whose goal is obtain a sample from an online frame
population, we can envision situations in which post-stratification might make a
bigger difference.
Since we partnered with Facebook directly, we are able to obtain an actual
probability sample of frame population members to invite to our survey.
In general, this may not always be possible and so
researchers may wish to try to recruit survey participants
using incomplete lists of mobile phone numbers, online ads, or other sources from which a probability sample is not
possible. In all cases, we recommend that researchers (i) critically assess the extent
to which their sampling mechanism produced a representative sample of the online frame population; and
(ii) consider using post-stratification or other calibration approaches to adjust systematic
differences between the sample and the frame population.

\begin{figure}[p]
    \centering
        \includegraphics[]{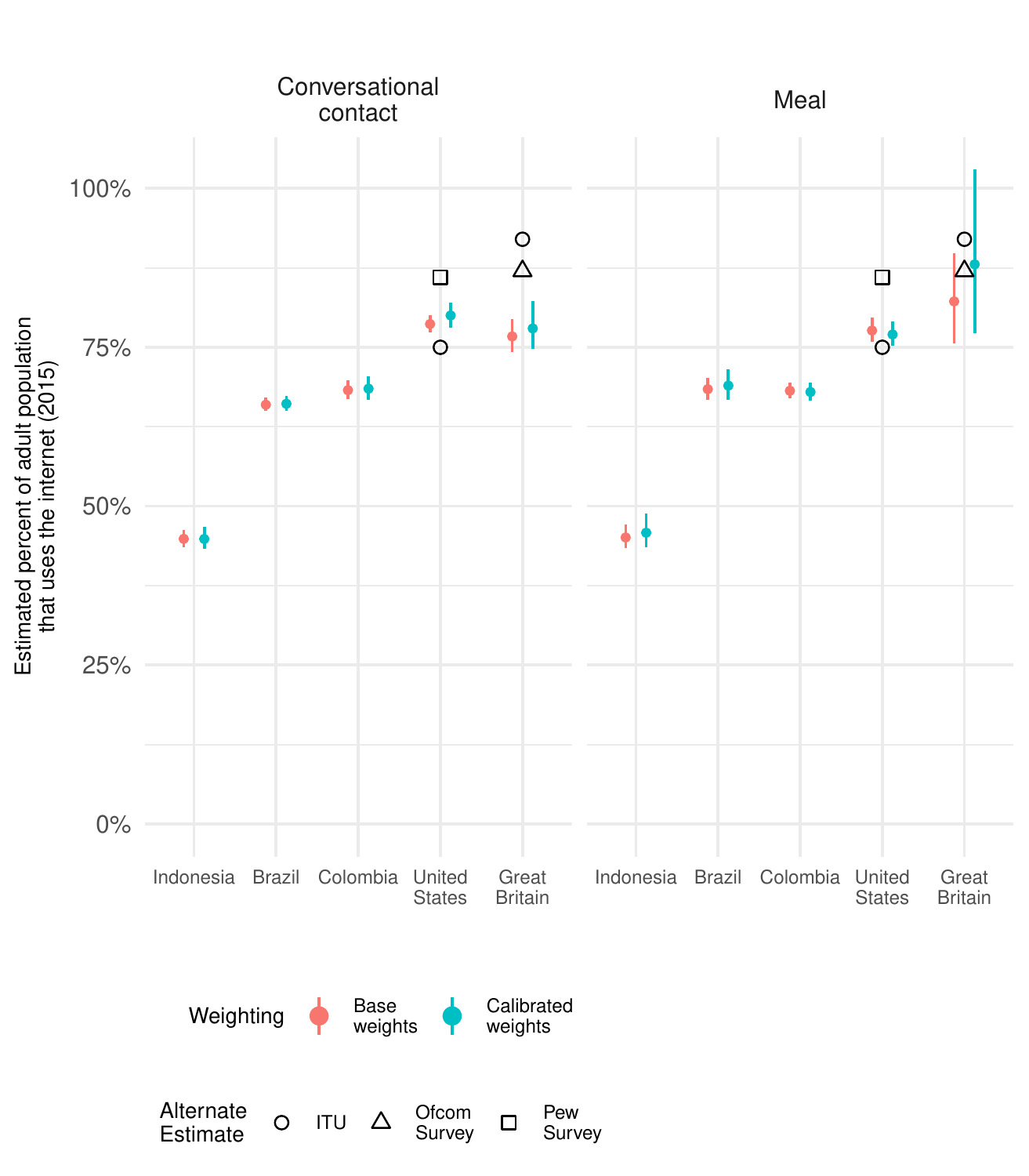}
        \caption{
        Estimated percentage of 2015 adult population that uses the internet, by country and for each of the two networks, with and without adjusting weights using post-stratification. The `base weights' are unadjusted and the `calibrated weights' are adjusted. Estimates are very similar to the ones shown in Figure \ref{fig:estimates}; the main difference is that the meal estimate for Great Britain is slightly lower in this figure.  }
        \label{fig:estimates-nocalib}
\end{figure}

\textbf{Sampling designs for future studies}

As we discuss above, we analyze our results as a simple random sample without
replacement, and our frame population---i.e., the group of people who were
eligible to be included as respondents in the survey---was people who actively
use Facebook in each of the five
countries we study. It was not necessary to use any stratification,
oversampling, or other features of more complex sample designs for our study.

Future researchers who wish to adopt our methodology may consider using
more complex sampling designs, including stratification, oversampling, and so
forth.
Since our results are derived using the design-based sampling framework, our
estimators and variance estimation approach will generalize naturally to complex
samples; a text on design-based sampling will have more details (e.g., Sarndal et al. 2003).

However, we note that modified sampling designs may require some attention
to the conditions that the estimator relies upon. In some cases, it may be appropriate to modify
the design used in our study to be more appropriate to the setting and the quantity of interest
(see also Appendix \ref{sec:simple-model}).
As an example, if researchers
wished to produce estimates of internet adoption in rural and urban areas
separately, then it would make sense to (i) stratify the sample by rural and
urban areas; and (ii) consider changing the survey question to focus on
respondents' network members who live in the same area they do.
That way, rural respondents would report about rural internet use while urban
respondents reported about urban internet use.
In this case, it would also be desireable for post-stratification to be conducted for
age/sex/urbanicity groups, rather than just age/sex groups.
Of course, other quantities of interest may require different changes to the exact design
of the survey; in the same way that no single design is appropriate for all household surveys,
it is also true that no single design will be appropriate for all online network reporting surveys.

\textbf{Survey instrument}

\begin{figure}[p]
    \centering
        \subfloat[]{\includegraphics[]{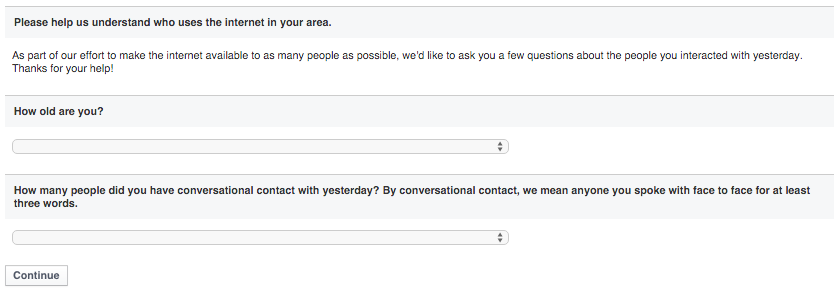}}\\
        \subfloat[]{\includegraphics[]{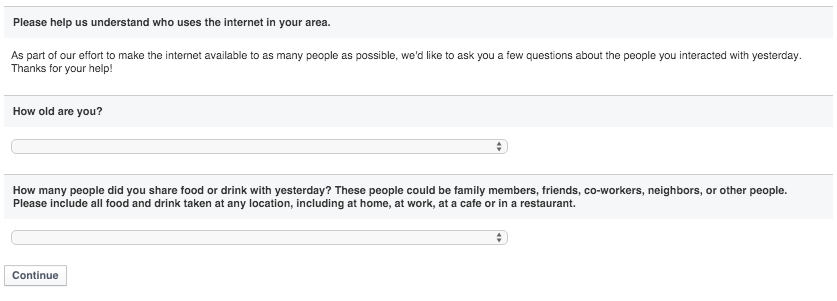}}
        \caption{
        First page of questions shown to survey respondents in (a) the conversational contact survey; and, (b) the meal survey.
        }
        \label{fig:survey-screenshot-0}
\end{figure}

\begin{figure}[p]
    \centering
    \subfloat[]{\includegraphics[width=5in]{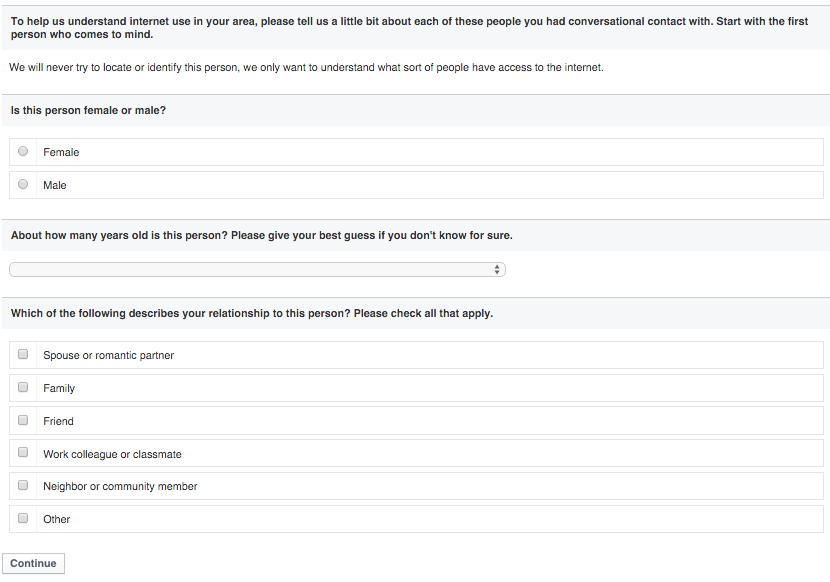}}\\
    \subfloat[]{\includegraphics[width=5in]{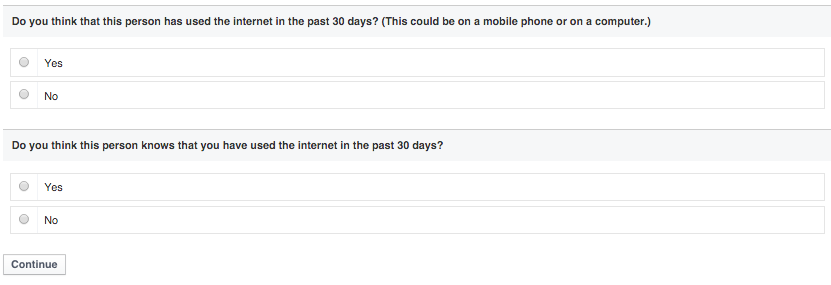}}\\
    \subfloat[]{\includegraphics[width=5in]{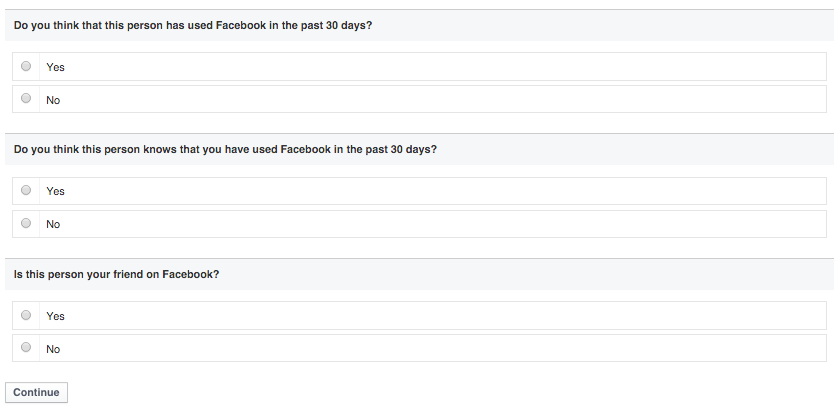}}
    \caption{Survey instrument used to ask about each detailed alter: panels (a), (b), and (c) show the first, second, and third screens asked about each of up to three detailed alters.}
    \label{fig:survey-screenshot-1}
\end{figure}

Figure \ref{fig:survey-screenshot-0} and Figure \ref{fig:survey-screenshot-1} show the
instruments that were used in our survey. (These instruments were translated
from English into the dominant local language for each country. Translated
instruments were back-translated into English as a quality check, as is standard
in survey research.)
Nothing about the survey instrument was particularly complex, and we expect that
delivering a similar survey would be possible using most existing systems.

Although we did not precisely measure the amount of time taken to complete each
survey, log data suggest that the average amount of time to complete the survey
was approximately 4.3 minutes for the conversational contact instrument and
approximately 3.8 minutes for the meal instrument.

\hypertarget{sec:generalized}{%
\section{Generalized estimates}\label{sec:generalized}}

Our approach relies upon survey respondents to be able to report whether or not members of their
personal networks used the internet in the last 30 days.
In reality, respondents may not be perfectly aware of their network members' internet use.
We see two ways to address this issue.
First, our sensitivity framework (Appendix \ref{sec:sensitivity-analysis}) can be used to assess how big an impact
errors in reporting will have on size estimates.
Second, for some groups, it may be possible to try to estimate the level of awareness of internet use
from the respondents themselves (Feehan and Salganik (2016a) has an in-depth analysis of this idea applied to the network scale-up estimator).
We now describe how we explored this second approach in more detail.

In our survey, we explored estimating visibility directly from the respondents by asking, for each detailed alter,
whether or not the detailed alter was aware of the respondent's internet use
(see the survey instrument in Figure \ref{fig:survey-screenshot-1}).
The idea is that the average extent to which detailed alters are aware of the
respondents' internet use can be used to approximate awareness of internet use in general.

Mathematically, an alternate approach to estimating the number of internet users in a given country is given by

\begin{equation}
\widehat{N}_H^{\text{gen}} = \frac{\sum_{i \in s} w_i \frac{d_i}{r_i} o_i}{\sum_{i \in s} w_i \frac{d_i}{r_i} z_i},
\label{eqn:gen-estimator}
\end{equation}

where the new term in the denominator, \(z_i\), is the total number of detailed alters reported by \(i\) who
are both on Facebook and reported to be aware of \(i\)'s internet use.
To understand Eq. \ref{eqn:gen-estimator} better, it is helpful to compare it to the estimator
used in the main text (Eq. \ref{eqn:estimator-any-weights}).
The estimator in Equation \ref{eqn:estimator-any-weights} has in its denominator
\(f_i\), which is the number of \(i\)'s detailed alters who are in the frame population.
We can re-write \(f_i\) as a sum of characteristics of each detailed alter \(f_i = \sum_{j \sim i} f_{ij}\), where \(f_{ij}\) is an indicator variable for whether or not \(i\)'s \(j\)th detailed
alter is reported to be on the frame population.

The generalized estimator in Equation \ref{eqn:gen-estimator} replaces \(f_i\) with another quantity, \(z_i\),
which is the total number of \(i\)'s detailed alters who are both in the frame population and who are reported
to be aware of \(i\)'s own internet use. Thus, we can write
\(z_i = \sum_{j \sim i} f_{ij} z_{ij}\),
where \(z_{ij}\) is an indicator for whether or not \(i\)'s \(j\)th detailed alter is reported to be aware of \(i\)'s internet use.
The intuition is that \(z_i\) will be lower than \(f_i\) when many of respondents' own network members are unaware
of the respondents' internet use, suggesting that general awareness about others' internet use is low.
Thus, as Feehan and Salganik (2016a) explains in more detail, the main idea is to use \(i\)'s perceived visibility to \(i\)'s network members as a way to approximate the visibility, rather than the degree, of people on the frame population.

In practice, we found that generalized estimates were almost indistinguishable from the
basic estimates with no adjustments; to illustrate this finding,
Figure \ref{fig:estimates-withgen} compares the generalized estimates to the basic estimates.
Thus, we focus on basic estimates in the main text.
However, other study designs, tie definitions, or quantities of interest may benefit more from using a generalized estimator, and we recommend that future studies continue to explore this possibility.

\begin{figure}[p]
    \centering
        \includegraphics[]{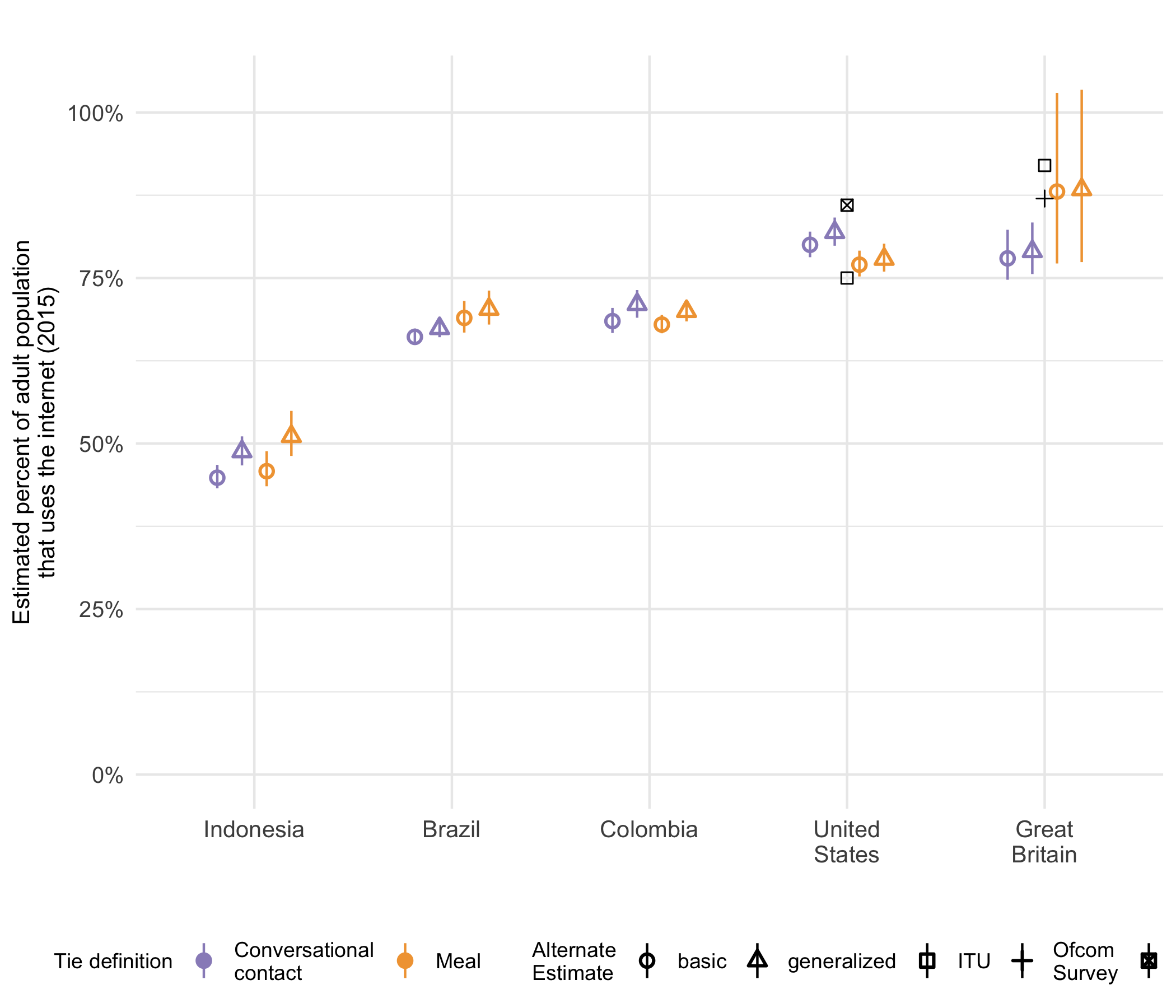}
        \caption{
        Comparison between basic and generalized estimates. Generalized estimates try to directly estimate visibility from the respondents, while basic estimates use the network size of survey respondents as an approximation of visibility. The figure shows that the two approaches produce very similar estimates; thus, we focus on basic estimates in the main text.
        }
        \label{fig:estimates-withgen}
\end{figure}

\hypertarget{sec:simple-model}{%
\section{A simple model to motivate the visibility estimator}\label{sec:simple-model}}

In the main text, the estimator we introduce is based upon approximating the visibility
of internet users, \(\bar{v}_{H,F}\) with the condition shown in Eq. \ref{eq:equal-tie-fbnofb},
which we repeat here for convenience:

\begin{equation}
\bar{d}_{H,F} = \bar{d}_{F,F}.
\label{eq:ap2-equal-tie-fbnofb}
\end{equation}

Eq. \ref{eq:ap2-equal-tie-fbnofb} says that the average number of connections
from an internet user to someone on Facebook equals the average number of connections
from someone on Facebook to someone else on Facebook.
In this Appendix, we show how a simple model can be used to derive the
relationship in Eq. \ref{eq:ap2-equal-tie-fbnofb}.
We note that this model is sufficient, but not
necessary: Eq. \ref{eq:ap2-equal-tie-fbnofb} is ultimately an empirical condition
that can hold even if the process that generated the underlying network is quite different
from this model; the goal of this model is to illustrate
one simple process through which Eq. \ref{eq:ap2-equal-tie-fbnofb} would be satisfied.

For the purposes of this Appendix, we will use the meal network as our
motivating example. We assume that reports are accurate, and we focus on
connections among people who use the internet, i.e.~members of group \(H\) (including
people who use Facebook and people who don't).
We ignore sampling and focus on population-level quantities, since the condition in Eq.
\ref{eq:ap2-equal-tie-fbnofb} is a population-level relationship.
And, in this appendix only, we will take
expectations with respect to our model, and not with respect to a sampling design (unlike the
rest of the paper).

The goal of the model is to make precise a situation in which people who are on the internet do not pay attention
to whether or not another internet user is on Facebook when deciding to share a meal together.
We call this situation \emph{homogenous mixing}.
Formally, suppose that meals are shared by each pair of internet users, \(i, j \in H\) with probability
\(p\), and that the probability that any pair of internet users shares a meal is independent of everyone
else's meal sharing.
Then we have an Erdos-Renyi random network with parameters \(N_H\) and \(p\).
An internet user will be connected to each of the \(N_F-1\) other internet users
with probability \(p\), so the expected degree of an internet user \(i \in H\) will be
\( \mathbb{E} [d_{i,F}] = (N_F - 1) p\).

Using the fact that \( \mathbb{E} [d_{i,F}] = (N_F-1)p\), and the independence of the edges,
we can now derive expressions for both sides of the relationship
in Eq. \ref{eq:ap2-equal-tie-fbnofb}. Starting with \(\bar{d}_{H,F}\), we find

\begin{equation}
\begin{aligned}
 \mathbb{E} [\bar{d}_{H,F}] &= \frac{1}{N_H}  \mathbb{E} [d_{H,F}]\\
&= \frac{1}{N_H}  \mathbb{E} [\sum_{i \in H} d_{i,F}]\\
&= \frac{1}{N_H} \sum_{i \in H} p (N_F-1)\\
&= (N_F-1)~p.
\end{aligned}
\end{equation}

And, for \(\bar{d}_{F,F}\), we get

\begin{equation}
\begin{aligned}
 \mathbb{E} [\bar{d}_{F,F}] &= \frac{1}{N_F}  \mathbb{E} [d_{F,F}]\\
&= \frac{1}{N_F} \sum_{i \in F}  \mathbb{E} [d_{i,F}]\\
&= \frac{1}{N_F} \sum_{i \in F} (N_F-1)p\\
&= (N_F-1) p.
\end{aligned}
\end{equation}

Thus, under this model, \(\bar{d}_{H,F} = \bar{d}_{F,F}\).

\hypertarget{extensions-of-the-model}{%
\subsection*{Extensions of the model}\label{extensions-of-the-model}}
\addcontentsline{toc}{subsection}{Extensions of the model}

We will briefly discuss two ways that this simple model can be elaborated:
first, we examine a situation in which the population is a mixture of different subgroups, each with
different levels of internet adoption and Facebook usage; within each subgroup, homogenous mixing holds.
Second, we examine a situation in which mixing is not homogenous.
We see pursuing these and other elaborations of the simple model as a useful direction for future work.

Before we can extend our model, however, we must introduce two useful facts.
The first fact describes how the average number of connections between two groups, \(A\) and \(B\),
can be decomposed when the first group \(A\) can be partitioned into two subgroups.\\
\hspace*{0.333em}\\
\hspace*{0.333em}

\begin{Fact}
\label{res:partition-dbar} Suppose that group \(A \subset U\) can be
partitioned into two subsets \(A_1\) and \(A_2\). Then for any group
\(B \subset U\), \begin{equation}
\bar{d}_{A, B} = p_1 \bar{d}_{A_1, B} + p_2 \bar{d}_{A_2,B},
\end{equation} where \(p_1 = \frac{|A_1|}{|A|}\) and
\(p_2 = \frac{|A_2|}{|A|}\).
\end{Fact}

Fact \ref{res:partition-dbar} follows from some algebra:
~\\
\hspace*{0.333em}

\begin{equation}
\begin{aligned}
\bar{d}_{A,B} 
&= \frac{1}{|A|} \left[ d_{A_1,B} + d_{A_2,B} \right]\\
&= \frac{1}{|A|} \left[ |A_1| \bar{d}_{A_1,B} + |A_2| \bar{d}_{A_2,B} \right]\\
&= \frac{|A_1|}{|A|} \bar{d}_{A_1,B} + \frac{|A_2|}{|A|} \bar{d}_{A_2,B}. 
\end{aligned}
\end{equation}

~\\
The second fact describes how the average number of connections between two groups,
\(A\) and \(B\), can be decomposed when the second group \(B\) can be partitioned into two subgroups.
~\\
\hspace*{0.333em}

\begin{Fact}
\label{res:subset-dbar} Suppose that group \(B \subset U\) can be
partitioned into two subsets \(B_1\) and \(B_2\). Then for any group
\(A \subset U\), \begin{equation}
\bar{d}_{A,B} = \bar{d}_{A,B_1} + \bar{d}_{A,B_2}.
\end{equation}
\end{Fact}

~\\
\hspace*{0.333em}\\
Fact \ref{res:subset-dbar} also follows from some algebra:

\begin{equation}
\begin{aligned}
\bar{d}_{A,B}
&= \frac{1}{|A|}\left[ d_{A,B_1} + d_{A,B_2}\right]\\
&= \bar{d}_{A,B_1} + \bar{d}_{A,B_2}.
\end{aligned}
\end{equation}

~\\
\hspace*{0.333em}

\hypertarget{a-non-interacting-mixture-of-different-populations}{%
\subsubsection*{A non-interacting mixture of different populations}\label{a-non-interacting-mixture-of-different-populations}}
\addcontentsline{toc}{subsubsection}{A non-interacting mixture of different populations}

Suppose the population can be partitioned into a set of groups such
that
(i) the number of internet users is different in each group;
(ii) the proportion of internet users who use Facebook is
different in each group; and
(iii) there are no meals shared between the groups and, within each group, mixing is homogenous.
We shall show that, in the the aggregate population, the relationship,
\(\bar{d}_{H,F} = \bar{d}_{F,F}\) will still hold.
As a motivating example, we will consider the case in which
there are rural and urban areas, and both internet adoption and Facebook usage
are higher in urban areas.

Formally, call the number of Facebook users in urban and rural areas \(N_{F \cap C}\) and \(N_{F \cap R}\);
call the number of internet users in urban and rural areas \(N_{H \cap C}\) and \(N_{H \cap R}\);
and let the fraction of internet users that is on Facebook be \(p_C = \frac{N_{F \cap C}}{N_{H \cap C}}\)
for urban areas
and \(p_R = \frac{N_{F \cap R}}{N_{H \cap R}}\) for rural areas.
Suppose that people only share meals with others who are in the same area as them;
so, there are no meals shared between rural and urban areas.
Finally, suppose that within these areas, homogenous mixing holds so that the condition in
Eq. \ref{eq:ap2-equal-tie-fbnofb} is satisfied;
that is, suppose that
\(\bar{d}_{H \cap R, F \cap R} = \bar{d}_{F \cap R,F \cap R}\) and
\(\bar{d}_{H \cap C,F \cap C} = \bar{d}_{F \cap C,F \cap C}\), where \(F \cap C\) and \(F \cap R\) are urban and rural Facebook users and \(H \cap C\) and \(H \cap R\) are urban and rural internet users.

The setup is a mixture of two populations, one urban and one rural. Within each
population, the simple model studied above describes the networks.
We wish to show that in the the aggregate population,
the relationship, \(\bar{d}_{H,F} = \bar{d}_{F,F}\) holds.

Applying Fact \ref{res:partition-dbar} to \(\bar{d}_{H,F}\), we have

\begin{equation}
\begin{aligned}
\bar{d}_{H,F} 
&= p_C \bar{d}_{H \cap C, F} + p_R \bar{d}_{H \cap R, F}\\
&= p_C \bar{d}_{H \cap C, F \cap C} + p_R \bar{d}_{H \cap R, F \cap R}.
\label{eqn:intermed-1}
\end{aligned}
\end{equation}

Where the last step follows from Fact \ref{res:subset-dbar} together with the
assumption that there are no meals shared between rural and urban areas.
Next, we apply the fact that mixing is homogenous within each area, i.e., that
\(\bar{d}_{H \cap R, F \cap R} = \bar{d}_{F \cap R, F \cap R}\)
and \(\bar{d}_{H \cap C, F \cap C} = \bar{d}_{F \cap C, F \cap C}\).
Eq. \ref{eqn:intermed-1} can thus be simplified to

\begin{equation}
\begin{aligned}
p_C \bar{d}_{H \cap C, F \cap C} + p_R \bar{d}_{H \cap R, F \cap R}
&= p_C \bar{d}_{F \cap C, F \cap C} + p_R \bar{d}_{F \cap R, F \cap R}\\
\label{eqn:intermed-2}
\end{aligned}
\end{equation}

Finally,
note that Fact \ref{res:subset-dbar}, together with the
assumption that there are no meals shared between rural and urban areas,
suggests substituting \(\bar{d}_{F \cap R, F \cap R} = \bar{d}_{F \cap R, F}\)
and \(\bar{d}_{F \cap U, F \cap U} = \bar{d}_{F \cap U, F}\)
into Eq. \ref{eqn:intermed-1}. This substitution produces

\begin{equation}
\begin{aligned}
p_C \bar{d}_{F \cap C, F \cap C} + p_R \bar{d}_{F \cap R, F \cap R}
&= p_C \bar{d}_{F \cap C, F} + p_R \bar{d}_{F \cap R, F}\\
&= \bar{d}_{F, F}.
\label{eqn:intermed-3}
\end{aligned}
\end{equation}

Thus, we have shown that \(\bar{d}_{H,F} = \bar{d}_{F,F}\) in this mixture model.

The analysis above emphasizes the fact that, in this simple model, the key condition is that
\(\bar{d}_{H,F} = \bar{d}_{F,F}\). If Facebook and internet adoption rates are different in different
types of places, but within a given place, Facebook users are evenly distributed through
the network, the condition in Eq. \ref{eq:ap2-equal-tie-fbnofb} may still be a reasonable basis for approximating visibility.

\hypertarget{non-uniform-mixing}{%
\subsubsection*{Non-uniform mixing}\label{non-uniform-mixing}}
\addcontentsline{toc}{subsubsection}{Non-uniform mixing}

In this section, we consider a second extension of the simple model: we generalize the model
to a situation in which people who use Facebook may interact differently with people who do
and who do not use Facebook.
Thus, we no longer assume that mixing is homogenous.
We will see that in this case, the condition \(\bar{d}_{H,F} = \bar{d}_{F,F}\) will no longer
hold in general.

In order to incorporate non-homogenous mixing between internet users who are and are not on
Facebook, we use a block model (see, e.g., Wasserman et al. 1994).
The block model says that the probability that \(i, j \in H\) are connected is a function only of
the `block' or group memberships of \(i\) and \(j\).
We consider a block model with two groups: Facebook users (\(F\)) and
internet users who do not use Facebook (\(H-F\)).
Every internet user is in one and only one of these two groups.
Figure \ref{fig:sb-matrix} shows the probability that \(i\) is connected to \(j\) given the group
memberships of \(i\) (rows) and \(j\) (columns) in terms of two parameters: \(\phi \in [0,1]\) and
\(\sigma \in [0,1]\).
The parameter \(\phi\) controls the probability of edges between two members of the same group,
and the parameter \(\sigma\) is a factor by which the probability of a connection is reduced
between two nodes who are in different groups.
When \(\sigma = 1\), this block model reduces to the simple homogenous mixing model we considered above.
On the other hand, when \(\sigma < 1\), there is non-homogenous mixing: someone on Facebook is
more likely to share a meal with someone else on Facebook than someone who is not on Facebook.

\begin{figure}
\begin{equation}
\bordermatrix{
                    & F               & H-F              \cr
    F            & \phi               & \sigma\cdot\phi  \cr
    H-F          & \sigma\cdot\phi    & \phi
}
\end{equation}
\caption{%
Matrix describing the block model: entry $(i,j)$ shows the probability that a randomly
chosen node from group $i$ is connected to a randomly chosen node from group $j$. 
The matrix is parameterized by $\phi \in [0,1]$, the probability of a connection between
two members of the same group; and $\sigma \in [0,1]$, the extent to which the probability of
connectivity is diminished when two nodes are in different groups.
When $\sigma = 1$, we have the homogenous mixing model considered above;
when $\sigma < 1$, there is non-homogenous mixing.
}
    \label{fig:sb-matrix}
\end{figure}

Now we will derive expressions for \(\bar{d}_{F,F}\) and \(\bar{d}_{H,F}\) under this
block model.
First, note that each node in \(F\) has a probability \(\phi\) of being connected to each of the
\(N_F-1\) other nodes in \(F\); thus, \( \mathbb{E} [\bar{d}_{F,F}] = (N_F-1)\phi\).
Also, note that if \(N_F\) is much bigger than 1, which it will typically be in the situations we are interested in,
then \( \mathbb{E} [\bar{d}_{F,F}] \approx N_F \phi\).

Next, we consider \(\bar{d}_{H,F}\).
By Fact \ref{res:partition-dbar}, this quantity
can be written as

\begin{equation}
\begin{aligned}
 \mathbb{E} [\bar{d}_{H,F}]
&=  \mathbb{E} [p_{F|H}~\bar{d}_{F,F} + (1-p_{F|H})\bar{d}_{H-F,F}] && \text{(by Fact \ref{res:partition-dbar})}\\
&= p_{F|H}~ \mathbb{E} [\bar{d}_{F,F}] + (1-p_{F|H}) \mathbb{E} [\bar{d}_{H-F,F}],
\end{aligned}
\label{eqn:intermed2-1}
\end{equation}

where we have written the proportion of internet users who also use Facebook as
\(p_{F|H} = \frac{N_F}{N_H}\).

Now, each node in \(H-F\) will have a probability \(\phi \sigma\) of being connected to
each of the \(N_F\) nodes in \(F\); thus, \( \mathbb{E} [\bar{d}_{H-F,F}] = N_F \phi \sigma\).
Substituting this relationship into Eq. \ref{eqn:intermed2-1}, along with the approximation
\( \mathbb{E} [\bar{d}_{F,F}] \approx N_F \phi\) discussed above, we obtain

\begin{equation}
\begin{aligned}
 \mathbb{E} [\bar{d}_{H,F}]
&= p_{F|H}~ \mathbb{E} [\bar{d}_{F,F}] + (1-p_{F|H}) \mathbb{E} [\bar{d}_{H-F,F}]\\
&\approx p_{F|H}~N_F \phi + (1-p_{F|H}) N_F \phi \sigma.
\end{aligned}
\label{eqn:intermed2-2}
\end{equation}

Comparing Eq. \ref{eqn:intermed2-2} to \(\bar{d}_{F,F} \approx N_F \phi\), we can see that the
two will not, in general be the same; indeed, as long as \(p_{F|H} < 1\) or \(\sigma < 1\), they
will be different.

To better understand this result, we can incorporate this analysis into our sensitivity framework.
Under perfect reporting, the parameter
\(\eta = \frac{\bar{v}_{H,F}}{\bar{v}_{F,F}}\)
can be written as
\(\eta = \frac{\bar{d}_{H,F}}{\bar{d}_{F,F}}\).
Now we can substitute the expressions we just derived to see what the value of \(\eta\) would be
under this block model:

\begin{equation}
\begin{aligned}
\eta 
&= \frac{\bar{v}_{H,F}}{\bar{v}_{F,F}} && \\
&= \frac{\bar{d}_{H,F}}{\bar{d}_{F,F}} && \text{(under perfect reporting)}\\
&= \frac{ p_{F|H}~(N_F-1) \phi + (1-p_{F|H}) N_F \phi \sigma }{ (N_F-1) \phi}\\
&\approx \frac{ p_{F|H}~N_F \phi + (1-p_{F|H}) N_F \phi \sigma }{ N_F \phi} && \text{(since $N_F-1\approx N_F$)}\\
&= p_{F|H} + (1-p_{F|H}) \sigma.
\end{aligned}
\label{eqn:eta-block}
\end{equation}

Eq. \ref{eqn:eta-block} shows that, under this block model, non-homogenous mixing will change
the parameter \(\eta\) so that it is different from one.
\(\eta\) will be farther from 1 when \(\sigma\) is farther from 1--that is, when there is more
non-homogenous mixing--and when the share of internet users that is on Facebook, \(p_{F|H}\),
is smaller.

\hypertarget{summary}{%
\subsection*{Summary}\label{summary}}
\addcontentsline{toc}{subsection}{Summary}

To recap, in this appendix, we introduced three simple models to help understand the condition
\(\bar{d}_{F,F} = \bar{d}_{H,F}\),
which is the basis of our approach to approximating the visibility of internet users in the
main text.
First, we saw that under a model in which people on the internet mix homogenously, paying no attention to
whether or not they are on Facebook, this relationship would be expected to hold.
Next, we saw that if the population is a mixture of different populations -- say, urban and rural people --
then the condition can also hold, as long as the different populations don't mix, and
mixing is homogenous within each population.
Finally, we introduced a third model in which there is non-homogenous mixing.
We saw that this third model could lead to values of the \(\eta\) parameter that are different from 1,
meaning that the condition is violated.

\hypertarget{sec:additional-results}{%
\section{Additional results}\label{sec:additional-results}}

Table \ref{tbl:raw_degree} reports estimated average network size (degree) for each tie definition.

\input{figures/pooled_average_degree.tex}

Fig. \ref{fig:alters} illustrates the detailed alters subsampled from each respondent's personal network.

\begin{figure}[t]
    \centering
        \subfloat[\label{fig:all-alters}]{\includegraphics[width=1.75in]{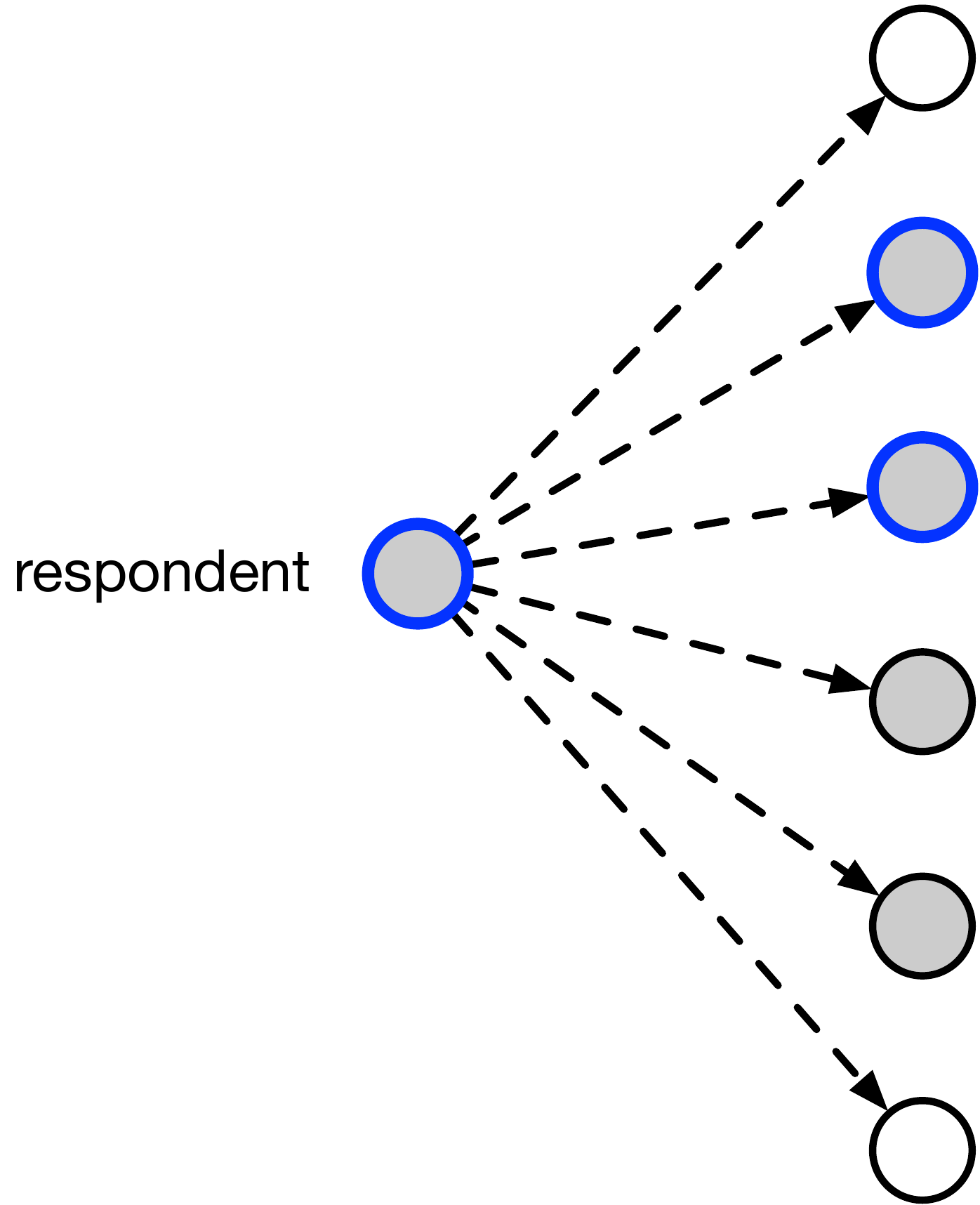}}
        \subfloat[\label{fig:detailed-alters}]{\includegraphics[width=1.75in]{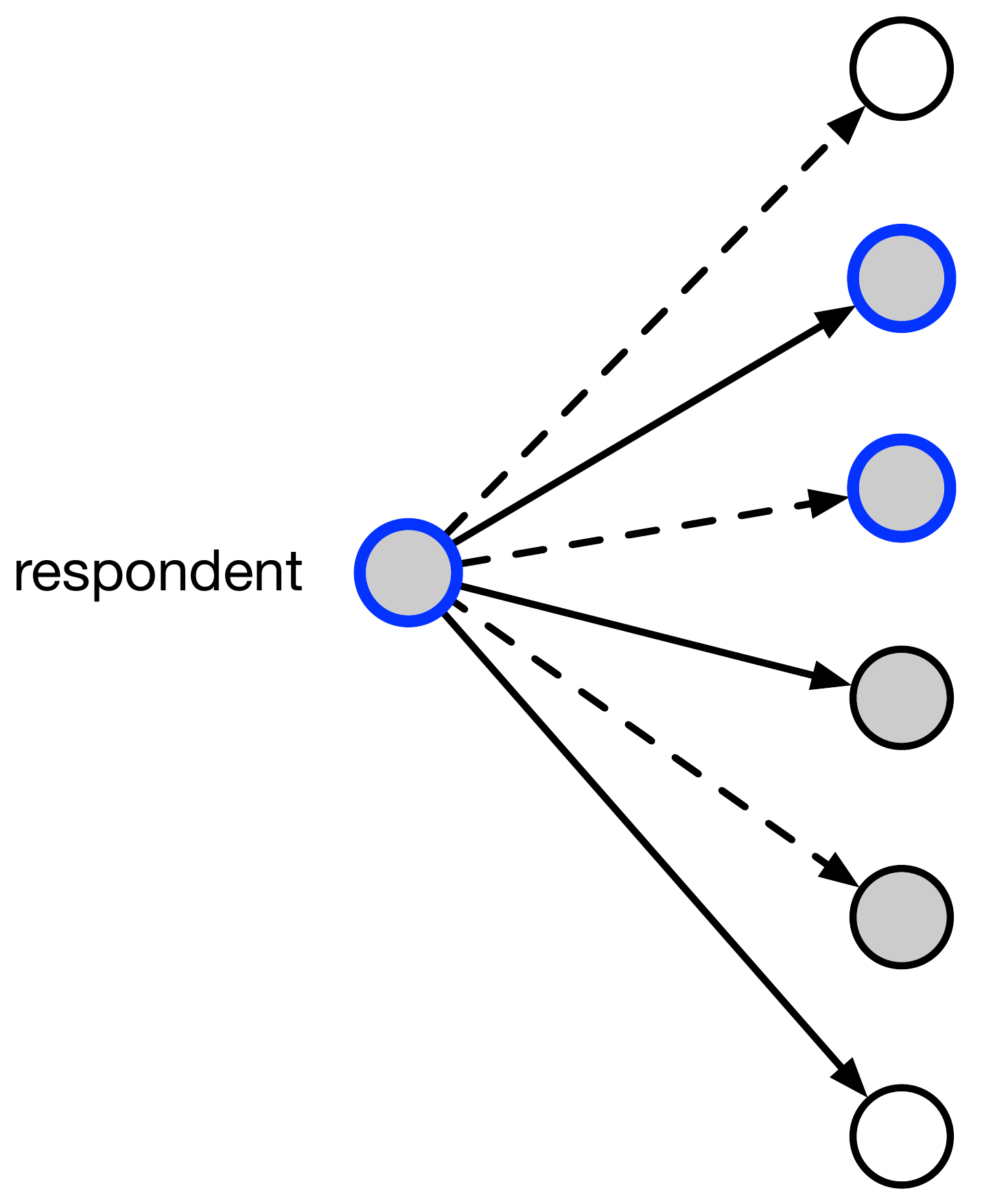}}
        \caption{(a) A survey respondent who is sampled online can be asked to report about members of one of her offline personal networks (e.g. her kin, friendship, or contact networks). Her responses contain information about both people who are online and people who are offline.
        (b) In order to reduce respondent burden, we asked for more detailed information about internet use, gender, and age for three \emph{detailed alters} in each respondent's personal network (solid lines).}
        \label{fig:alters}
\end{figure}

Table \ref{tbl:ic_err} provides a summary of the comparison of TAE between the two tie definitions
within each country (the information that is visualized in Fig.~\ref{fig:ic-diff}).

\input{figures/pooled_ic_err.tex}

~\\
\hspace*{0.333em}\\
\hspace*{0.333em}

\clearpage

\end{document}

%% file: figures/pooled_average_degree.tex
\begin{table}[ht]
\centering
\begin{tabular}{lll}
  \hline
country & Conversational contact & Meal \\ 
  \hline
Brazil & 13.1 (12.5, 13.6) & 6.3 (5.9, 6.6) \\ 
  Colombia & 10.5 (10, 11.1) & 7.2 (6.9, 7.6) \\ 
  Great Britain & 12.7 (11.6, 13.9) & 4.4 (3.7, 5.3) \\ 
  Indonesia & 11 (10.4, 11.6) & 7.5 (7, 8) \\ 
  United State & 12.1 (11.6, 12.5) & 5 (4.6, 5.4) \\ 
   \hline
\end{tabular}
\caption{Estimated average degree and 95\% confidence interval, by type of personal network} 
\label{tbl:raw_degree}
\end{table}

%% file: figures/pooled_ic_err.tex
\begin{table}[ht]
\centering
\begin{tabular}{lll}
  \hline
Country & Median TAE & Mean TAE (95\% CI) \\ 
  \hline
Brazil & 8.62 & 8.64 (2.16, 15.54) \\ 
  Colombia & 4.77 & 4.82 (-1.36, 11.14) \\ 
  Great Britain & 5.43 & 5.67 (-2.58, 14.22) \\ 
  Indonesia & -6.71 & -6.85 (-24.13, 8.82) \\ 
  United States & 1.39 & 1.4 (-3.28, 6.52) \\ 
   \hline
\end{tabular}
\caption{Estimated mean and 95\% CI for the TAE, the difference in total absolute error in internal consistency checks across all age groups for the conversational contact network minus the same quantity for the meal network.  Positive values mean that the conversational contact network was less internally consistent than the meal network, as measured by absolute error. } 
\label{tbl:ic_err}
\end{table}